\documentclass{article}
\usepackage{amsmath,amssymb,amsthm}
\usepackage{hyperref}

\newtheorem{theorem}{Theorem}
\newtheorem{definition}{Definition}

\newtheorem{lemma}{Lemma}
\newtheorem{corollary}{Corollary}

\newcommand{\ket}[1]{\left|#1\right\rangle}
\newcommand{\bra}[1]{\left\langle #1\right|}

\newcommand{\F}{\mathbb{F}_2}

\newcommand{\N}{\mathbb{N}}
\newcommand{\C}{\mathbb{C}}

\newcommand{\orderedpairs}[1]{Pairs \left(#1\right)}

\newcommand{\twobytwo}[4]{\begin{pmatrix} #1 & #2 \\ #3 & #4\end{pmatrix}}

\newcommand{\threebyone}[3]{\begin{pmatrix} #1 \\ #2 \\ #3\end{pmatrix}}

\newcommand{\id}[1]{I_{#1}}
\newcommand{\revdiag}[1]{\Omega_{#1}}
\newcommand{\e}[1]{e_{#1}}
\newcommand{\g}[1]{G_{#1}}
\renewcommand{\L}[2]{\mathcal{L}_{#1}\left(#2\right)}
\newcommand{\B}[2]{\mathcal{B}_{#1}\left(#2\right)}
\newcommand{\tl}[1]{T^{Left}\left(#1\right)}
\newcommand{\tr}[1]{T^{Right}\left(#1\right)}
\newcommand{\tm}[1]{T^{moves}\left(#1\right)}
\newcommand{\tmr}[1]{T^{rev}\left(#1\right)}
\newcommand{\ttcr}[1]{T^{tcr}\left(#1\right)}

\renewcommand{\v}[1]{\mathcal{V}_{#1}}

\newcommand{\x}[1]{X_{#1}}
\newcommand{\z}[1]{Z_{#1}}
\newcommand{\h}[1]{H_{#1}}

\renewcommand{\c}[1]{W_{#1}}
\newcommand{\s}[1]{S_{#1}}

\newcommand{\q}[1]{Q_{#1}}
\newcommand{\qubit}[1]{q_{#1}}

\newcommand{\pauligroup}[1]{\mathcal{P}_{#1}}
\newcommand{\cliffordgroup}[1]{\mathcal{C}_{#1}}
\newcommand{\symplecticgroup}[1]{\mathcal{SP}_{#1}}
\newcommand{\paulihomomorphism}[1]{\Theta_{#1}}
\newcommand{\cliffordhomomorphism}[1]{\Psi_{#1}}

\newcommand{\paci}[2]{\left\langle #1,#2 \right\rangle}

\newcommand{\erasurechannel}[2]{\mathcal{E}_{#1}\left(#2\right)}
\newcommand{\depolarizingchannel}[2]{\mathcal{D}_{#1}\left(#2\right)}

\newcommand{\prob}[1]{\mathbb{P}\left(#1\right)}
\newcommand{\expect}[1]{\mathbb{E}\left(#1\right)}
\newcommand{\indicator}[1]{\mathbb{I}\left(#1\right)}

\newcommand{\orderfunction}{\xi}

\newcommand{\bincdf}[1]{\mathcal{F}\left(#1\right)}
\newcommand{\bincdfinv}[1]{\mathcal{F}^{-1}\left(#1\right)}

\newcommand{\normalcdf}[1]{\Phi\left(#1\right)}
\newcommand{\normalcdfinv}[1]{\Phi^{-1}\left(#1\right)}

\newcommand{\floor}[1]{\left\lfloor #1 \right\rfloor}
\newcommand{\ceiling}[1]{\left\lceil #1 \right\rceil}
\newcommand{\parenth}[1]{\left( #1 \right)}

\newcommand{\curlybr}[1]{\left\{ #1 \right\}}
\newcommand{\absval}[1]{\left| #1 \right|}

\newcommand{\fracexp}[3]{\parenth{\frac{#1}{#2}}^{#3}}
\newcommand{\logpar}[1]{\log_2\left(#1\right)}

\newcommand{\relent}[2]{D\left(#1\middle\|#2\right)}

\newcommand{\paulich}[1]{\mathcal{N}_{#1}}

\newcommand{\cosetrate}[1]{R_{coset}\left(#1\right)}
\newcommand{\errorguessrate}[1]{R_{errorguess}\left(#1\right)}

\newcommand{\errorguessrateachievability}[1]{R_{errorguess}^{achievability}\left(#1\right)}

\newcommand{\errorguessrateconverse}[1]{R_{errorguess}^{converse}\left(#1\right)}
\newcommand{\coseterrorprob}[1]{\varepsilon_{coset}\left(#1\right)}
\newcommand{\errorguesserrorprob}[1]{\varepsilon_{errorguess}\left(#1\right)}

\newcommand{\errorguesserrorprobconverse}[1]{\varepsilon_{errorguess}^{converse}\left(#1\right)}

\newcommand{\errorguesserrorprobachievability}[1]{\varepsilon_{errorguess}^{achievability}\left(#1\right)}

\newcommand{\simplex}[1]{\Delta(#1)}
\newcommand{\discretesimplex}[2]{\Delta_{#2}(#1)}
\newcommand{\type}[1]{\mathbf{p}_{#1}}
\newcommand{\typeclass}[1]{\mathbf{T}(#1)}

\newcommand{\supportu}{S}
\newcommand{\supportv}{T}

\begin{document}
\title{Canonical Form and Finite Blocklength Bounds for Stabilizer Codes}
\author{Dimiter Ostrev}
\date{}
\maketitle

\begin{abstract}
First, a canonical form for stabilizer parity check matrices of arbitrary size and rank is derived. Next, it is shown that the closely related canonical form of the Clifford group can be computed in time $O(n^3)$ for $n$ qubits, which improves upon the previously known time $O(n^6)$. Finally, the related problem of finite blocklength bounds for stabilizer codes and Pauli noise is studied. A finite blocklength refinement of the hashing bound is derived, and it is shown that no argument that uses guessing the error as a substitute for guessing the coset can lead to a significantly better achievability bound. 
\end{abstract}

\subsubsection*{IEEE publication notice}

This article has been accepted for publication in IEEE Transactions on Information Theory. This is the author's version which has not been fully edited and content may change prior to final publication. Citation information: DOI 10.1109/TIT.2026.3688060

\subsubsection*{IEEE copyright notice}

\copyright 2026 IEEE. All rights reserved, including rights for text and data mining and training of artificial intelligence and similar technologies. Personal use of this material is permitted. Permission from IEEE must be obtained for all other uses, in any current or future media, including reprinting/republishing this material for advertising or promotional purposes, creating new collective works, for resale or redistribution to servers or lists, or reuse of any copyrighted component of this work in other works.

\subsubsection*{Acknowledgement}

The first version of this work was completed while Dimiter Ostrev was with the Institute of Communications and Navigation, German Aerospace Center (DLR) and was supported by the Bavarian Ministry of Economic Affairs, Regional Development and Energy through the project High-Efficiency Stabilizer Codes. The author is now with the Interdisciplinary Centre for Security, Reliability and Trust, University of Luxembourg, where he updated and expanded the article and obtained the new results in Section \ref{sec:n_independent_identical_channels}. 

\section{Introduction}

Stabilizer codes \cite{gottesman1996class,calderbank1997quantum}, quantum analogues of classical linear codes, are a widely studied method for protecting quantum states from noise. This article considers canonical forms and finite blocklength bounds for stabilizer codes. 

Matrix decompositions and canonical forms have long been important in mathematics. In quantum error correction, various decompositions of Clifford group elements have been proposed, with applications ranging from compilation of quantum circuits to generation of uniformly random Clifford group elements; see for example \cite{duncan2020graphtheoretic,bravyi2021hadamardfree,bassler2023synthesisof,khesin2025universalgraphrepresentation} and references therein for more details. 

In another recent work \cite{ostrev2024quantum}, a joint matrix decomposition of a pair of orthogonal matrices plays an important role in the design and analysis of the proposed family of quantum Calderbank-Shor-Steane codes. The codes in \cite{ostrev2024quantum} have a set of linearly dependent stabilizer measurements, and therefore can be used to correct errors on the qubits and errors in the syndrome simultaneously.  The matrix decomposition in \cite{ostrev2024quantum} gives information about both the quantum CSS code and the classical linear code that protects against errors in the syndrome. 

It would be desirable to generalize the matrix decomposition of \cite{ostrev2024quantum} from CSS to arbitrary stabilizer codes. This is achieved in the present article. Theorem \ref{thm:canonical_form_for_stabilizer_pcm} below gives a canonical form for stabilizer parity check matrices of arbitrary size and rank. 

The main idea in the proof of existence of the matrix decomposition, both in \cite{ostrev2024quantum} and in Theorem \ref{thm:canonical_form_for_stabilizer_pcm} in the present article, is to use a variant of Gaussian elimination with both row and column operations, with the additional restriction that the column operations are chosen so as to preserve the symplectic products of the rows at each step. Thus, the algorithm for computing the matrix decomposition is a close relative of the disentangling algorithms for the Clifford group, such as the one described in \cite[Theorem 10.6 and Exercise 10.40]{nielsen2010quantum} and more recently in \cite[Section II]{bravyi2021hadamardfree}. 

To prove uniqueness of the matrix decomposition for stabilizer codes, the present article extends and simplifies the methods used previously in \cite{bravyi2021hadamardfree} to establish a canonical form for the Clifford group.  \cite{bravyi2021hadamardfree} describes a collection of subgroups of the Borel group, parametrized by elements of the Weyl group.  This collection of subgroups is then used to describe the canonical form and prove its uniqueness. See also \cite[Section V]{maslov2018shorterstabilizer} for a more detailed description of the Borel subgroup, the Weyl subgroup, and the Bruhat decomposition of the symplectic group.

The approach in \cite{bravyi2021hadamardfree} is not directly applicable to stabilizer parity check matrices; the obstacle is roughly speaking that the latter are not invertible and need not even have linearly indepedent rows. To overcome the difficulty, the present article defines a more general family of groups of lower triangular matrices. These groups are shown to arise naturally during Gaussian elimination, and are used to describe and prove uniqueness of the canonical form for stabilizer parity check matrices. 

An approach centered on Gaussian elimination has various theoretical and practical advantages. Besides yielding a canonical form for stabilizer parity check matrices, it also gives new insights into the canonical form for the Clifford group. 

For example, \cite[Section II]{bravyi2021hadamardfree} gave a two-stage algorithm to compute the canonical form of the Clifford group. The first stage, a disentangling algorithm closely related to Gaussian elimination, takes time $O(n^3)$ for $n$ qubits. However, the second stage takes time $O(n^6)$, so the overall runtime to compute the canonical form is $O(n^6)$. 

The present article gives a more precise analysis of the Gaussian elimination stage. In Theorem \ref{thm:clifford_group_canonical_form} below, it is shown that Gaussian elimination produces the canonical form directly, without the need for post-processing. Thus, the overall runtime is reduced from $O(n^6)$ to $O(n^3)$, which roughly squares the number of qubits for which it is practically feasible to compute the canonical form. 

Moreover, the present article gives a simplified description of the relevant subgroups and self-contained proofs. The exposition here relies only on the structure of Gaussian elimination, rather than on more advanced topics such as Bruhat decomposition and $(B,N)$-pairs. 

Insights in the structure of stabilizer codes may be applied to further problems. For example, the canonical form in Theorem \ref{thm:canonical_form_for_stabilizer_pcm} gives an encoding circuit for a given stablizer parity check matrix. The canonical form can also be used to generate a uniformly random stabilizer parity check matrix of a given size and rank by an algorithm that uses the optimal number of random bits. 

These applications are not considered in detail here. The argument for the encoding circuit is similar to \cite[Section III]{ostrev2024quantum}, where a decomposition of CSS parity check matrices is used to derive the encoding circuit. The algorithm for generation of uniformly random stabilizer parity check matrices of given size and rank is similar to the algorithm in \cite[Section III]{bravyi2021hadamardfree} for the generation of uniformly random Clifford group elements. 

Instead, this article considers the following question: given Pauli noise affecting a finite number $n$ of qubits, what is the highest rate of a stabilizer code that keeps the error probability below a target $\epsilon$, and what is the lowest error probability of a stabilizer code of a given rate $r$? 

Theorem \ref{thm:general_bounds} below gives a finite blocklength refinement of the hashing bound, applicable to arbitrary $n$-qubit Pauli noise. Then, Theorems \ref{thm:erasure_channel_bounds} and \ref{thm:depolarizing_channel_bounds} show how the general bound can be computed efficiently in the important special cases of $n$ independent qubit erasure channels and $n$ independent qubit depolarizing channels. Theorem \ref{thm:computing_the_bounds} shows that the bounds can be computed in polynomial time in the case of $n$ independent identical Pauli channels. 

Given an achievability bound, it is natural to ask whether it is in some sense optimal. In classical error correction, there is a well developed theory of finite blocklength bounds \cite{polyanskiy2010channel}. In particular, in the special cases of the classical erasure and binary symmetric channels, the finite blocklength achievability bounds come with nearly matching finite blocklength converse bounds. 

Unfortunately, the corresponding problem in quantum error correction is less well understood. \cite{ashikhmin2014fidelitylowerbounds} focuses on achievability bounds for stabilizer and CSS codes on the depolarizing channel, but does not give any converse bounds. \cite{tomamichel2016quantum} derives both achievability and converse bounds applicable to general quantum error correcting codes, not just stabilizer codes. However, the achievability and converse bounds of \cite{tomamichel2016quantum} nearly match only in the case of the qubit dephasing channel, which is closely related to the classical binary symmetric channel, and in the case of the quantum erasure channel with classical post processing, which is closely related to the classical erasure channel. On the other hand, for the depolarizing channel, there is a constant size gap between the achievability and converse bounds of \cite{tomamichel2016quantum}, which corresponds to the gap between the hashing lower and dephasing upper bounds on the capacity of the depolarizing channel. 

In this context, the present article makes the following contribution: it shows that the finite blocklength refinement of the hashing bound in Theorems \ref{thm:erasure_channel_bounds} and \ref{thm:depolarizing_channel_bounds} is nearly optimal among a certain class of arguments: those that use guessing the most likely error from the syndrome as a substitute for guessing the most likely stabilizer coset from the syndrome. 

Such a result is important because the relaxation from guessing the coset to guessing the error is widely used in the literature on stabilizer codes. In particular, the previous finite blocklength achievability bound of \cite{ashikhmin2014fidelitylowerbounds} for stabilizer codes on the depolarizing channel also uses this relaxation. 

The rest of the article is structured as follows: Section \ref{sec:preliminaries_and_notation} covers preliminaries and some background material. Section \ref{sec:groups_of_lower_triangular_matrices} introduces a family of groups of lower triangular matrices, and Section \ref{sec:canonical_forms} uses these groups to establish canonical forms for unrestricted and stabilizer parity check matrices and for the symplectic and Clifford group. Section \ref{sec:finite_blocklength_bounds} gives the finite blocklength refinement of the hashing bound and the proof that it is nearly optimal among the class of arguments that substitute guessing the error for guessing the coset. Section \ref{sec:conclusion} concludes the article and gives some directions for future work. 

\section{Preliminaries and notation}\label{sec:preliminaries_and_notation}

This section introduces notation that will be used later on, and covers some background material on quantum stabilizer codes and Pauli channels. Some prior acquaintance with quantum computation and quantum information is assumed; see, for example, \cite{nielsen2010quantum}. 

After some basic notation (subsection \ref{sec:basic_notation}), the Pauli group (subsection \ref{sec:pauli_group}) and Clifford group (subsection \ref{sec:clifford_group}) are introduced, along with the associated homomorphisms to vectors and matrices over the field with two elements. Next, the class of Pauli channels is described in subsection \ref{sec:pauli_channels}; these channels apply a random Pauli error and possibly give some classical side information correlated with the error. This class of channels has a number of useful properties: it is closed under conjugation by the Clifford group (subsection \ref{sec:conjugation_by_clifford}), closed under preparation and measurement of a subset of the qubits in the computational basis (subsection \ref{sec:preparation_and_measurement}), and closed under Pauli corrections conditional on the side information (subsection \ref{sec:conditional_correction}). Moreover, the diamond distance of a Pauli channel from the identity can be easily computed (subsection \ref{sec:diamond_distance}). The preceding observations show that the search for the optimal stabilizer code to correct Pauli noise admits a convenient combinatorial description (subsection \ref{sec:coset_decoding}). A relaxation of this combinatorial problem is given in the final subsection \ref{sec:guessing_the_error}. 

\subsection{Some notation}\label{sec:basic_notation}

Let $[n]$ denote the finite set $\{1,\dots,n\}$. Let $\F$ denote the field with two elements and $\F^n$ denote the space of column vectors with $n$ components from $\F$. For $i \in [n]$, let $\e{n,i}$ denote the $i$-th standard basis vector of $\F^n$. Let $\F^{m \times n}$ denote the space of $m \times n$ matrices over $\F$. Let $\id{n}$ denote the $n \times n$ identity, and let $\revdiag{n}$ denote the $n\times n$ reverse diagonal matrix
\begin{equation}
\revdiag{n}=\sum_{i=1}^n \e{n,i}\e{n,n+1-i}^{\top}
\end{equation} 

\subsection{The Pauli group}\label{sec:pauli_group}

Consider a quantum system with $n$ qubits. Let $\x{n,i},\z{n,i}$ denote the Pauli X and Z operations acting on the $i$-th qubit, and let
\begin{equation}
\q{2n}=(\x{n,1},\dots,\x{n,n},\z{n,n},\dots,\z{n,1})
\end{equation}
be the $2n$-tuple of single qubit Pauli X and Z operations ordered in a particular way. The reason for this particular ordering will become clear later. Individual components of $\q{2n}$ will be denoted by $\q{2n,i}$; for example, $\q{2n,1}=\x{n,1}$ and $\q{2n,2n}=\z{n,1}$ are the Pauli X and Z operations on the first out of $n$ qubits. 

The Pauli group on $n$ qubits $\pauligroup{n}$ is generated by the components of $\q{2n}$ and by $iI_{2^n}$. Two elements $P,P'$ of the Pauli group either commute or anti-commute; this will be denoted as follows:
\begin{equation}
\paci{P}{P'}=\begin{cases} 1 & \text{if } PP'+P'P=0 \\ 0 & \text{if } PP'-P'P=0 \end{cases}
\end{equation}
The operation $\paci{}{}$ has the following properties
\begin{align}
\paci{P}{P'}&=\paci{P'}{P}\\
\paci{PP'}{P''}&=\paci{P}{P''}\oplus\paci{P'}{P''}\\
\paci{P}{P'}&=\paci{UPU^{-1}}{UP'U^{-1}}
\end{align}
for any $P,P',P''\in\pauligroup{n}$ and any unitary $U$. 

Finally, the map $\paulihomomorphism{n}:\pauligroup{n}\rightarrow\F^{2n}$ given by 
\begin{equation}
\paulihomomorphism{n}(P)=\threebyone{\paci{\q{2n,2n}}{P}}{\vdots}{\paci{\q{2n,1}}{P}}
\end{equation}
is a surjective group homomorphism with kernel $\{\pm I, \pm i I\}$. $\paulihomomorphism{n}$ sends $\q{2n,i}$ to $\e{2n,i}$. Additinally, the identity
\begin{equation}
\forall P,P' \in \pauligroup{n},\paci{P}{P'}=\paulihomomorphism{n}(P)^{\top} \revdiag{2n}\paulihomomorphism{n}(P')
\end{equation}
holds. 

\subsection{The Clifford group}\label{sec:clifford_group}

The Clifford group on $n$ qubits consists of unitary matrices that map Pauli group elements to Pauli group elements under conjugation:
\begin{equation}
\cliffordgroup{n}=\{W \in \C^{2^n \times 2^n}:WW^\dagger=I 
\text{ and }\forall P\in\pauligroup{n}, WPW^{-1} \in \pauligroup{n}\}
\end{equation}
The symplectic group consists of matrices in $\F^{2n \times 2n}$ that preserve the symplectic form $x^{\top}\revdiag{2n}y$:
\begin{equation}
\symplecticgroup{2n}=\left\{C\in\F^{2n\times 2n}:C^{\top}\revdiag{2n}C=\revdiag{2n}\right\}
\end{equation}
The map $\cliffordhomomorphism{n}:\cliffordgroup{n}\rightarrow\symplecticgroup{2n}$ defined by
\begin{equation}
\forall i,j\in[2n],
\e{2n,i}^{\top} \cliffordhomomorphism{n}(W) \e{2n,j} = \paci{\q{2n,2n+1-i}}{W\q{2n,j}W^{-1}}
\end{equation}
is a surjective group homomorphism with kernel $\{cP:c\in\C, P\in\pauligroup{n}\}$. Additional useful properties of $\cliffordhomomorphism{n}$ are 
\begin{align}
\cliffordhomomorphism{n}(W^{-1})=\revdiag{2n}\cliffordhomomorphism{n}(W)^{\top}\revdiag{2n} \\
\forall P \in \pauligroup{n},\paulihomomorphism{n}(WPW^{-1})=\cliffordhomomorphism{n}(W)\paulihomomorphism{n}(P)
\end{align}

\subsection{Pauli channels}\label{sec:pauli_channels}

\begin{definition}
An $n$ qubit Pauli channel with side information (or just Pauli channel) is specified by the joint distribution $p_{UV}$ of a random variable $U$ taking values in $\F^{2n}$ and another discrete random variable $V$. The channel specified by the joint distribution $p$ will be denoted by $\paulich{p}$; it transforms $n$ qubit states as follows:  
\begin{equation}\label{eq:generic_pauli_channel}
\rho \mapsto \paulich{p}(\rho)=\sum_{u,v} p_{UV}(u,v) \q{2n}^u\rho\parenth{\q{2n}^u}^\dagger \otimes \ket{v}\bra{v}
\end{equation}
where $\q{2n}^u=\q{2n,1}^{u_1}\q{2n,2}^{u_2}\dots \q{2n,2n}^{u_{2n}}$. 
\end{definition}

Examples of Pauli channels are the qubit erasure channel and the qubit depolarizing channel. The qubit erasure channel with parameter $\delta$ transforms the state $\rho$ of a single qubit to 
\begin{equation}\label{eq:erasure_channel}
\erasurechannel{\delta}{\rho}=(1-\delta)\rho \otimes \ket{0}\bra{0} 
+\frac{\delta}{4}\left(\rho + \x{}\rho\x{}+\z{}\rho\z{}+\x{}\z{}\rho\z{}\x{}\right)\otimes\ket{1}\bra{1}
\end{equation}
The side information indicates whether an erasure has occurred or not. The qubit depolarizing channel with parameter $\delta$ transforms the state $\rho$ of a single qubit to
\begin{equation}\label{eq:depolarizing_channel}
\depolarizingchannel{\delta}{\rho}=(1-\delta) \rho + \frac{\delta}{3}\left(\x{}\rho\x{}+\z{}\rho{}\z{}+\x{}\z{}\rho\z{}\x{}\right)
\end{equation}
In this case, there is no side information. 

\subsection{Conjugation of a Pauli channel by a Clifford group element}\label{sec:conjugation_by_clifford}

Conjugating a Pauli channel by a Clifford unitary gives another Pauli channel. Specifically, if \eqref{eq:generic_pauli_channel} is conjugated by $W\in\cliffordgroup{n}$ the resulting channel transforms $n$-qubit states as follows:   
\begin{multline}\label{eq:clifford_conjugated_pauli_channel}
\rho \mapsto \sum_{u,v} p_{UV}(u,v) W \q{2n}^u W^{-1} \rho W \parenth{\q{2n}^u}^\dagger W^{-1} \otimes \ket{v}\bra{v} \\
=\sum_{u,v} p_{UV}(u,v)  \q{2n}^{\cliffordhomomorphism{n}(W)u}  \rho  \parenth{\q{2n}^{\cliffordhomomorphism{n}(W)u}}^\dagger  \otimes \ket{v}\bra{v} 
\end{multline}

\subsection{Preparation and measurement}\label{sec:preparation_and_measurement}

Consider now the $n$-qubit Pauli channel specified by \eqref{eq:clifford_conjugated_pauli_channel}. Suppose the $n$ qubits are divided in two registers: the first $m$ qubits, and the remaining $k=n-m$ qubits. Suppose further that the first $m$ qubits are prepared in the zero state, then the channel is applied, then the first $m$ qubits are measured in the computational basis. These operations transform the $n$-qubit Pauli channel into a $k$-qubit Pauli channel acting on the last $k$ qubits. The channel on the last $k$ qubits transforms $k$-qubit states as follows:  
\begin{multline}\label{eq:effective_logical_channel}
\rho\mapsto\sum_{u,v} p_{UV}(u,v) \\
* \q{2k}^{\sum_{i=1}^{2k} \e{2k,i}\e{2n,m+i}^{\top} \cliffordhomomorphism{n}(W)u} \rho \parenth{\q{2k}^{\sum_{i=1}^{2k}\e{2k,i}\e{2n,m+i}^{\top}\cliffordhomomorphism{n}(W)u}}^\dagger \\ \otimes \ket{v,\sum_{i=1}^m \e{m,i}\e{2n,i}^{\top}\cliffordhomomorphism{n}(W)u}\bra{v,\sum_{i=1}^m \e{m,i}\e{2n,i}^{\top}\cliffordhomomorphism{n}(W)u}
\end{multline}

\subsection{Conditional correction}\label{sec:conditional_correction}

Finally, a Pauli correction is applied on the remaining qubits to transform \eqref{eq:effective_logical_channel} into a channel that is close to the identity. To simplify notation, rewrite \eqref{eq:effective_logical_channel} as
\begin{equation}
\rho\mapsto\sum_{u',v'}p'_{U'V'}(u',v')\q{2k}^{u'}\rho\parenth{\q{2k}^{u'}}^\dagger \otimes \ket{v'}\bra{v'}
\end{equation}
A conditional correction specified by some conditional probabilities $\hat{p}_{\hat{U}'|V'}$, followed by tracing out the side information result in the channel
\begin{equation}\label{eq:channel_after_correction}
\rho\mapsto\sum_{u',v',\hat{u}'}p'_{U'V'}(u',v')\hat{p}_{\hat{U}'|V'}(\hat{u}'|v')\q{2k}^{u'+\hat{u}'}\rho\parenth{\q{2k}^{u'+\hat{u}'}}^\dagger 
\end{equation}

\subsection{Diamond distance from the identity}\label{sec:diamond_distance}

The diamond distance \cite[Section 9.1.6]{wilde2017quantum} between the channel in equation \eqref{eq:channel_after_correction} and the identity channel is $2\prob{\hat{U}'\neq U'}$. Indeed, the triangle inequality shows that the diamond distance is at most $2\prob{\hat{U}'\neq U'}$, and using a maximally entangled state in the maximization problem for the diamond distance \cite[Theorem 9.1.1]{wilde2017quantum} shows that it is at least $2\prob{\hat{U}'\neq U'}$. 

\subsection{Optimal stabilizer codes and coset decoding}\label{sec:coset_decoding}

Given a finite blocklength $n$ and an $n$-qubit Pauli channel $\paulich{p}$, it is desirable to understand the highest rate of a stabilizer code that can correct the noise $\paulich{p}$ while keeping the probability of error below a given $\epsilon$, and the lowest probability of error of a stabilizer code of a given rate $r$. The preceeding discussion shows that these can be defined combinatorially as follows:

\begin{definition}\label{def:coset_decoding}
Let $p_{UV}$ be a joint probability distribution of a random vector $U$ in $\F^{2n}$ and another discrete random variable $V$. 

For $C\in\symplecticgroup{2n}$, partition $n=m+k$ of the qubits and decoding algorithm $D$, the probability of error is  
\begin{equation}\label{eq:prob_of_error_for_coset_decoding}
\prob{D\parenth{V,\sum_{i=1}^m \e{m,i}\e{2n,i}^{\top}CU}\neq \sum_{i=1}^{2k} \e{2k,i}\e{2n,m+i}^{\top} CU}
\end{equation}

For $\epsilon>0$, let $\cosetrate{p_{UV},\epsilon}$ be the maximum value of $k/n$ such that there exist $C,D$ with probability of error at most $\epsilon$. 

For $r=k/n$, let $\coseterrorprob{p_{UV},r}$ be the lowest error probability of a pair $C,D$ that uses $m=n-k$ bits of syndrome. 
\end{definition}

Note that $\sum_{i=1}^m \e{m,i}\e{2n,2n+1-i}^{\top}CU$ is irrelevant to the distance from the identity channel, and therefore it is also not mentioned in the combinatorial formulation in Definition \ref{def:coset_decoding}. Thus, the Pauli error $U$, about which $m$ bits are observed as the syndrome and $2k$ bits are guessed by the decoder, is determined only up to a coset of a vector space of dimension $m$. 

\subsection{Guessing the error from the syndrome}\label{sec:guessing_the_error}

Analyzing the optimal rate and error probability for coset decoding presents considerable difficulties. Therefore, it is worthwhile to introduce a relaxation of this problem, and to consider the optimal rate and error probability for the related task of guessing the error from the syndrome. 

\begin{definition}
Consider a probability distribution $p_{UV}$ as before. 

For $C\in\symplecticgroup{2n}$, partition $n=m+k$ of the qubits and decoding algorithm $D$, the probability of incorrect guess of $U$ from the syndrome is 
\begin{equation}
\prob{D\parenth{V,\sum_{i=1}^m \e{m,i}\e{2n,i}^{\top} CU}\neq U}
\end{equation}

For $\epsilon>0$, let $\errorguessrate{p_{UV},\epsilon}$ be the maximum value of $k/n$ such that there is a pair $C,D$ with the probability of guessing $U$ incorrectly at most $\epsilon$. 

For $r=k/n$, let $\errorguesserrorprob{p_{UV},r}$ be the lowest probability of guessing $U$ incorrectly of a pair $C,D$ that uses $m=n-k$ bits of syndrome. 
\end{definition}

Since a pair $C,D$ that can guess $U$ can be adapted to coset decoding with equal or lower error probability, the optimal rates and error probabilities satisfy
\begin{align}
\errorguessrate{p_{UV},\epsilon} &\leq \cosetrate{p_{UV},\epsilon} \\
\coseterrorprob{p_{UV},r} & \leq \errorguesserrorprob{p_{UV},r}
\end{align}

\section{Groups of lower triangular matrices}\label{sec:groups_of_lower_triangular_matrices}

Subsection \ref{sec:family_of_groups} introduces a family of groups of lower triangular matrices. Subsection \ref{sec:generators_unrestricted_case} shows that each of these groups is generated by a suitable subset of the Gaussian move matrices. Subsection \ref{sec:intersection_with_symplectic_group} considers the intersection of these groups with the symplectic group and Subsection \ref{sec:generators_symplectic_case} shows that these intersections are generated by suitable symplectic analogues of the Gaussian moves.

\subsection{A family of groups from transitive sets of pairs}\label{sec:family_of_groups}

\begin{definition}
For $n \in \N$, let 
\begin{equation}
\orderedpairs{n}=\left\{(i,j):i,j\in[n],i>j\right\}
\end{equation}
\end{definition}

\begin{definition}
A subset $T$ of $\orderedpairs{n}$ is called transitive if
\begin{equation}
(i,j)\in T,(j,k)\in T \implies (i,k) \in T
\end{equation}
\end{definition}

\begin{definition}
To a transitive subset $T$, associate the set of lower triangular matrices
\begin{equation}
\L{n}{T}=\id{n}+span\left\{\e{n,i}\e{n,j}^{\top}:(i,j)\in T\right\}
\end{equation}
\end{definition}

\begin{lemma}
$\L{n}{T}$ is closed under matrix multiplication and matrix inverse. 
\end{lemma}

\begin{proof}
Take any 
\begin{equation}
A=\sum_{(i,j)\in T} a_{i,j}\e{n,i}\e{n,j}^{\top}, \;\;\; B=\sum_{(k,l)\in T} b_{k,l} \e{n,k}\e{n,l}^{\top}
\end{equation}
From the transitive property, deduce $AB \in span\left\{\e{n,i}\e{n,j}^{\top}:(i,j)\in T\right\}$. Then, 
\begin{equation}
(\id{n}+A)(\id{n}+B)=\id{n}+A+B+AB \in \L{n}{T}
\end{equation}
Moreover, $A$ is nilpotent. Let $d$ be the largest integer such that $A^{2^d} \neq 0$. Then, 
\begin{equation}
(\id{n}-A)^{-1}=(I+A)(I+A^2)(I+A^4)\dots (I+A^{2^d})
\end{equation}
is also an element of $\L{n}{T}$. 
\end{proof}

Some examples:
\begin{enumerate}
\item The lower triangular group is $\L{n}{\orderedpairs{n}}$. 

\item Matrices with off-diagonal entries only in certain rows and columns: take subsets $R,C$ of $[n]$. Then, $T=\orderedpairs{n} \cap (R\times C)$ is transitive. $\L{n}{T}$ is a group of lower triangular matrices whose non-zero off-diagonal entries appear only in rows in $R$ and columns in $C$.

\item The inversions and non-inversions of a permutation: let $\pi$ be a permutation of $[n]$. Let
\begin{align}
inv(\pi)=\left\{(i,j):i>j,\pi(i)<\pi(j)\right\}\\
ninv(\pi)=\left\{(i,j):i>j,\pi(i)>\pi(j)\right\}
\end{align}
Both $inv(\pi)$ and $ninv(\pi)$ are transitive. The associated groups $\L{n}{inv(\pi)}$ and $\L{n}{ninv(\pi)}$ played a role in the canonical form for invertible $n \times n$ matrices \cite{bravyi2021hadamardfree}. 
\end{enumerate}

\subsection{Generators and canonical form of the groups $\L{n}{T}$}\label{sec:generators_unrestricted_case}

For $i\neq j \in [n]$, let 
\begin{equation}
\g{n,i,j}=\id{n}+\e{n,i}\e{n,j}^{\top} \in \F^{n\times n}
\end{equation}
denote the corresponding Gaussian move matrix over $\F$. Left multiplication by $\g{n,i,j}$ adds row $j$ to row $i$. Right multiplication by $\g{n,i,j}$ adds column $i$ to column $j$. $\g{n,i,j}\g{n,j,i}\g{n,i,j}$ acts on the left as a row swap and on the right as a column swap.  

Consider the identity:
\begin{lemma}\label{lemma:ClassicalLowerTriangularIdentity}
Let $\left\{a_{i,j}:(i,j)\in\orderedpairs{n}\right\}$ be any collection of elements of $\F$ indexed by $\orderedpairs{n}$. Then,
\begin{equation}
\id{n}+\sum_{(i,j)\in\orderedpairs{n}} a_{i,j}\e{n,i}\e{n,j}^{\top} = \prod_{(i,j)\in\orderedpairs{n}} \g{n,i,j}^{a_{i,j}}
\end{equation} 
where in the product, the order of terms from left to right is 
\begin{equation}\label{eq:OrderOfLowerTriangularGaussianMoves}
(2,1), \dots, (n,1), (3,2), \dots, (n,2), \dots, (n,n-1)
\end{equation}
\end{lemma}

\begin{proof}
Think of the product as a polynomial in the standard basis vectors and their transposes. The ordering is such that all terms of degree 4 and higher vanish, and only the terms of degree 2 and the identity remain. 
\end{proof}

This identity gives a set of generators and a canonical form for all the groups $\L{n}{T}$. 

\begin{corollary}
Take $n \in \N$ and transitive $T \subset \orderedpairs{n}$. The group $\L{n}{T}$ is generated by
\begin{equation}
\left\{\g{n,i,j}:(i,j)\in T\right\}
\end{equation}
Each element of $\L{n}{T}$ can be uniquely written as a product 
\begin{equation}
\prod_{(i,j)\in T} \g{n,i,j}^{a_{i,j}}
\end{equation}
where the product respects the order \eqref{eq:OrderOfLowerTriangularGaussianMoves}
\end{corollary}

\subsection{Intersection of the groups $\L{2n}{T}$ with the symplectic group}\label{sec:intersection_with_symplectic_group}

\begin{definition}
A subset $T$ of $\orderedpairs{2n}$ will be called closed under reversal if
\begin{equation}
(i,j) \in T \iff (2n+1-j,2n+1-i) \in T
\end{equation}
\end{definition}

\begin{definition}
To a subset $T$ that is both transitive and closed under reversal, associate the group 
\begin{equation}
\B{2n}{T}=\L{2n}{T} \cap \symplecticgroup{2n}
\end{equation}
\end{definition}

\subsection{Generators and canonical form of the groups $\B{2n}{T}$}\label{sec:generators_symplectic_case}

\subsubsection{The Clifford Gaussian moves}

For $i \neq j \in \{1, \dots, 2n\}$ let
\begin{equation}
\c{2n,i,j}=\begin{cases}\sqrt{\q{2n,i}}& \text{if } i+j=2n+1\\ \frac{1}{2}(I+ \q{2n,2n+1-j} & \\+ \q{2n,i} - \q{2n,2n+1-j} \q{2n,i}) & \text{otherwise}\end{cases}
\end{equation}
These unitaries will be called the Clifford Gaussian moves, because they play a role analogous to the classical Gaussian moves $\g{n,i,j}$. 

The Clifford gaussian moves are also related to the standard Clifford one and two qubit gates:
\begin{enumerate}
\item For $i\in [n]$, $\c{2n,i,2n+1-i}$ is a phase gate conjugated by a Hadamard gate on qubit $i$. 
\item For $i\in [n]$, $\c{2n,2n+1-i,i}$ is a phase gate on qubit $i$. 
\item For $i\in [n]$, $\c{2n,2n+1-i,i}\c{2n,i,2n+1-i}\c{2n,2n+1-i,i}=\sqrt{\z{n,i}}\sqrt{\x{n,i}}\sqrt{\z{n,i}}=\frac{1+i}{\sqrt{2}}\h{n,i}$ is a Hadamard gate on qubit $i$. 
\item For $i \neq j \in [n]$, $\c{2n,i,j}=\c{2n,(2n+1-j),(2n+1-i)}$ is a CNOT gate with control qubit $j$ and target qubit $i$.
\item For $i \neq j \in [n]$, $\c{2n,i,j}\c{2n,j,i}\c{2n,i,j}$ is a SWAP of qubits $i,j$. 
\item For $i > n \geq j$, $i+j\neq2n+1$, $\c{2n,i,j}=\c{2n,2n+1-j,2n+1-i}$ is a CZ gate on qubits $2n+1-i,j$. 
\item For $i\leq n<j$, $\c{2n,i,j}=\c{2n,2n+1-j,2n+1-i}$ is a CZ gate conjugated by Hadamard gates on qubits $i,2n+1-j$. 
\end{enumerate}

\subsubsection{The symplectic Gaussian moves} 

The symplectic Gaussian moves are the images of the Clifford Gaussian moves under the homomorphism $\cliffordhomomorphism{n}$: for $i \neq j \in \{1, \dots, 2n\}$ let
\begin{equation}
\s{2n,i,j}=\cliffordhomomorphism{n}(\c{2n,i,j}) 
=
\begin{cases}
\id{2n}+\e{2n,i}\e{2n,j}^{\top} & \text{if } i+j=2n+1\\
\id{2n}+\e{2n,i}\e{2n,j}^{\top} & \\
+\e{2n,2n+1-j}\e{2n,2n+1-i}^{\top} & \text{otherwise}
\end{cases}
\end{equation}

\subsubsection{Clearing entire rows and columns}

It is well-known that there are simple combinations of classical Gaussian moves that can be used to clear entire rows and columns. It is less obvious that there are combinations of the more complicated Clifford and symplectic Gaussian moves that have similar properties. To illustrate the similarities and differences of the classical and symplectic case, both are now given explicitly. 

\emph{With classical Gaussian moves:} Take $i\in[n]$. Take a vector 
\begin{equation}
v=\sum_{j=1}^n v_j \e{n,j} \in \F^n
\end{equation}
such that the $i$-th component is 0 (i.e. $v_i=0$). Let
\begin{equation}
\g{n,v,i}=\id{n}+v\e{n,i}^{\top}=\prod_{j \neq i}\g{n,j,i}^{v_j}
\end{equation}
and let $\g{n,i,v^{\top}}=\g{n,v,i}^{\top}$.

\emph{With symplectic Gaussian moves:} Take $i\in[2n]$. Take a vector
\begin{equation}
v=\sum_{j=1}^{2n} v_j \e{2n,j} \in \F^{2n}
\end{equation}
such that the $i$-th component is 0 (i.e. $v_i=0$). Let
\begin{multline}\label{eq:SymplecticGaussianWholeColumnMoves}
\s{2n,v,i}
=\id{2n}+v\e{2n,i}^{\top}+\revdiag{2n}\e{2n,i}v^{\top}\revdiag{2n}+v_{2n+1-i}\e{2n,2n+1-i}\e{2n,i}^{\top}\\
=\s{2n,2n+1-i,i}^{\sum_{j=1}^n v_jv_{2n+1-j}}\prod_{j\neq i}\s{2n,j,i}^{v_j}
\end{multline}
and let $\s{2n,i,v^{\top}}=\s{2n,v,i}^{\top}$. Similar combinations of Clifford gates appeared previously in \cite[Lemma 4]{bravyi2021hadamardfree}. 

Some relevant properties of $\s{2n,v,i}$ are:
\begin{lemma}\label{lemma:PropertiesOfSymplecticMoves}
\begin{enumerate}
\item For fixed $i$ the symplectic Gaussian moves $\{\s{2n,j,i}:j \in[2n]\backslash \{i\}\}$ commute.
\item The two different expressions for $\s{2n,v,i}$  in equation \eqref{eq:SymplecticGaussianWholeColumnMoves}, one with the identity and outer products, the other with the symplectic Gaussian moves, are equal.  
\item $\s{2n,v,i}^2=\id{2n}$.
\item $\s{2n,v,i}$ is supported only on the main diagonal, column $i$ and row $2n+1-i$. 
\item $\s{2n,v,i} \e{2n,i} = \e{2n,i} + v$
\item If $\e{2n,i}^{\top}u=v^{\top}\revdiag{2n}u=0$, then $\s{2n,v,i}u=u$. 
\end{enumerate}
\end{lemma}

\begin{proof}
From the definitions using simple calculations. The most work is required to show that the two expressions in equation \eqref{eq:SymplecticGaussianWholeColumnMoves} are equal. To do this, consider the action of the two expressions on the standard basis vectors, exploiting the commutativity from part 1 and the fact that $\s{2n,i,j}$ acts as the identity on all standard basis vectors except $\e{2n,j},\e{2n,2n+1-i}$. 
\end{proof}

\subsubsection{An identity reveals the generators and canonical form of the groups $\B{2n}{T}$}

In the classical case, the generators and canonical form of the groups $\L{n}{T}$ were obtained from the identity in Lemma \ref{lemma:ClassicalLowerTriangularIdentity}. A symplectic analogue of that identity will now be given. 

Because of the symplectic constraint, the off-diagonal entries of matrices in $\L{2n}{\orderedpairs{2n}}\cap\symplecticgroup{2n}$ cannot all be chosen freely. The following linear function focuses on a set of entries that can be chosen freely, and that together determine the rest. 

\begin{definition}
For $n\in\N$, let $\v{n}:\F^{2n\times 2n}\rightarrow \F^{2n \times n}$ be the linear function
\begin{equation}
\v{n}(B) = \sum_{j=1}^n \sum_{i=j+1}^{2n+1-j} \e{2n,i}\e{2n,i}^{\top} B \e{2n,j}\e{n,j}^{\top}
\end{equation}
\end{definition}

Now, consider the identity:
\begin{theorem}\label{thm:SymplecticLowerTriangularIdentity}
For every $B \in \B{2n}{\orderedpairs{2n}}$, 
\begin{equation}
B=\s{2n,\v{n}(B)\e{n,1},1}\s{2n,\v{n}(B)\e{n,2},2}\dots\s{2n,\v{n}(B)\e{n,n},n}
\end{equation}
\end{theorem}

\begin{proof}
Define a sequence of matrices by 
\begin{equation}
B_0=B, B_k = \s{2n,\v{n}(B)\e{n,k},k} B_{k-1}, k=1, \dots, n
\end{equation}
It will be shown by induction on $k$ that the matrices $B_k$ satisfy:
\begin{equation}
\e{2n,i}^{\top} B_k \e{2n,j} = \begin{cases}
\e{2n,i}^{\top}B\e{2n,j} & \text{if } i,j \in \{k+1,\dots, 2n-k\} \\
\e{2n,i}^{\top} \e{2n,j} & \text{otherwise}
\end{cases}
\end{equation}
For $B_0=B$, the claim holds. Suppose the claim holds for $B_{k-1}$. Note that 
\begin{multline}
B_k=\s{2n,\v{n}(B)\e{n,k},k}B_{k-1}\\
=B_{k-1}+\v{n}(B)\e{n,k}\e{2n,k}^{\top} B_{k-1} 
+ \e{2n,2n+1-k}\e{n,k}^{\top}\v{n}(B)^{\top}\revdiag{2n}B_{k-1} \\
+\e{2n,2n+1-k}\e{2n,2n+1-k}^{\top}B\e{2n,k}\e{2n,k}^{\top} B_{k-1}
\end{multline}
Since $\e{2n,k}^{\top}B_{k-1}=\e{2n,k}^{\top}$, the second term is supported only on the $k$-th column. The third and fourth terms are supported only on row $2n+1-k$. Thus, $B_k-B_{k-1}$ is supported only on column $k$ and row $2n+1-k$. Moreover, column $k$ of $B_{k-1}$ is $\e{2n,k}+\v{n}(B)\e{n,k}$, so column $k$ of $B_k$ is just $\e{2n,k}$. Finally, row $2n+1-k$ of $B_k$ is determined by column $k$ via the symplectic constraint:
\begin{equation}
\e{2n,2n+1-k}^{\top} B_k=\e{2n,k}^{\top} \revdiag{2n} B_k = \e{2n,k}^{\top} B_k^{\top} \revdiag{2n} B_k 
 = \e{2n,k}^{\top} \revdiag{2n} = \e{2n,2n+1-k}^{\top}
\end{equation}
This completes the induction. Then, $B_n=\id{2n}$, which proves the theorem. 
\end{proof}

The following Theorem gives properties of the groups $\B{2n}{T}$ related to generators and a canonical form:

\begin{theorem}
Let $T \subset \orderedpairs{2n}$ be both transitive and closed under reversal. Then, 
\begin{enumerate}
\item $(i,j) \in T$ implies $\s{2n,i,j} \in \B{2n}{T}$. 
\item If $i \in [2n]$ and $v\in \F^{2n}$ are such that $v$ is supported on positions $\{j: (j,i)\in T\}$, then $\s{2n,v,i} \in \B{2n}{T}$. 
\item $\{\s{2n,i,j}: j \in [n], j+1 \leq i \leq 2n+1-j, (i,j)\in T\}$ generates $\B{2n}{T}$.
\item The identity in Theorem \ref{thm:SymplecticLowerTriangularIdentity} provides a canonical expression of each $B \in \B{2n}{T}$ in terms of the generators. 
\end{enumerate}
\end{theorem}

\begin{proof}
Part 1 follows directly from the definitions. 

Part 2 follows from part 1, equation \eqref{eq:SymplecticGaussianWholeColumnMoves} and the following claim, which is used to deal with the correction term $\s{2n,2n+1-i,i}^{\sum_{j=1}^n v_jv_{2n+1-j}}$ in equation \eqref{eq:SymplecticGaussianWholeColumnMoves}:
\begin{equation}
\text{Claim: }(j,i) \in T \wedge (2n+1-j,i) \in T \implies (2n+1-i,i) \in T
\end{equation}
Indeed, closure under reversal implies $(2n+1-i,2n+1-j) \in T$ and then transitivity implies $(2n+1-i,i)\in T$. 

Parts 3 and 4 follow from part 2 and Theorem \ref{thm:SymplecticLowerTriangularIdentity}. 
\end{proof}

\section{Canonical forms for unrestricted and stabilizer parity check matrices, and for the symplectic and Clifford groups}\label{sec:canonical_forms}

The preparatory subsection \ref{sec:CanonicalFormForUnrestrictedMatrices} introduces the main ideas in the simpler case of unrestricted matrices. It also establishes notation and some Lemmas. Next, subsection \ref{sec:canonical_form_for_stabilizer_pcm} gives the canonical form for stabilizer parity check matrices. Finally, subsection \ref{sec:canonical_form_for_symplectic_group} considers the canonical form for the symplectic and Clifford groups. 

\subsection{Gaussian elimination produces a canonical form for matrices of a given size}\label{sec:CanonicalFormForUnrestrictedMatrices}

An application of the groups $\L{n}{T}$ will now be described. Groups of this type will be seen to arise naturally during Gaussian elimination. They will be used to show that the associated matrix decomposition is in fact a canonical form. 

One of the many variants of Gaussian elimination will now be recalled in some detail, with special attention on which row and column operations can potentially be used, as a function of the pivot positions. 

The search for pivots must proceed in some definite order; different orders result in different but analogous matrix decompositions \cite{strang2015algebra}. The present article uses the order "left and down": the search for pivots starts in the first row from the last element to the first, then in the second row from the last element to the first, etc. This choice produces decompositions with lower triangular matrices both on the left and on the right. 

The positions of the pivots can be described using a pair of functions. Let $\alpha:[r]\rightarrow [m]$ be increasing, and let  $\beta:[r] \rightarrow [n]$ be injective; to these functions correspond the $r$ pivot positions $(\alpha(i), \beta(i)),i=1,\dots,r$. These $r$ pivot positions can also be summarized in a matrix: let
\begin{equation}
\Pi(\alpha,\beta)=\sum_{i=1}^r\e{m,\alpha(i)}\e{n,\beta(i)}^{\top}
\end{equation}
$\Pi(\alpha,\beta)$ may be called an incomplete permutation matrix, because it has at most one 1 in each row and column, but may be non-square and may have zero rows and columns. 

During Gaussian elimination, the pivot positions are not revealed all at once, but step-by-step. Take $s>r$, and suppose that increasing $\alpha':[s]\rightarrow [m]$ and injective $\beta':[s]\rightarrow [n]$ agree with $\alpha,\beta$ on $[r]$. Then, call $(\alpha',\beta')$ an $s$-extension of $(\alpha,\beta)$, and call $(\alpha,\beta)$ the $r$-predecessor of $(\alpha',\beta')$. A given collection of pivot positions can have many possible extensions to a given higher rank, but it has a unique predecessor at a given lower rank. 

The algorithm progresses by taking a partially reduced matrix and simplifying it further. For $\alpha,\beta$ as above, call $A \in \F^{m \times n}$ $(\alpha,\beta)$-partially reduced if rows $1, \dots, \alpha(r)$ and columns $\beta(1),\dots,\beta(r)$ of $A$ coincide with the corresponding rows and columns of $\Pi(\alpha,\beta)$. If in addition $A = \Pi(\alpha,\beta)$, then $A$ is fully reduced, and the algorithm ends. 

If $A$ is partially but not fully reduced, then extend $(\alpha,\beta)$ to $(\alpha',\beta')$ so that $\alpha'(r+1),\beta'(r+1)$ is the position of the next pivot in the left and down order. Then, transform $A$ to a matrix $A'$ that is $(\alpha',\beta')$-partially reduced:
\begin{equation}\label{eq:GaussianEliminationStep}
A'=\left( \prod_{i=\alpha'(r+1)+1}^m \g{m,i,\alpha'(r+1)}^{\e{m,i}^{\top} A \e{n,\beta'(r+1)}} \right) A 
\left( \prod_{j=1}^{\beta'(r+1)-1} \g{n,\beta'(r+1),j}^{\e{m,\alpha'(r+1)}^{\top} A \e{n,j}} \right)
\end{equation}

Now, observe that on the left, only $\g{m,i,\alpha'(r+1)}$ with $i>\alpha'(r+1)$ are used. Similarly, on the right, only $\g{n,\beta'(r+1),j}$ with $j<\beta'(r+1)$ are used. However, on the right, there is an additional restriction: since columns $\beta(1), \dots, \beta(r)$ of $A$ are already reduced, $\g{n,\beta'(r+1),j}$ with $j \in Im(\beta)$ are not used.\footnote{If the search for pivots proceeded along columns rather than along rows, i.e. in the order last column top to bottom, then second last column top to bottom, etc., the additional restriction would be on the left rather than on the right.} 

These observations prompt the following definitions. 
\begin{definition}
For increasing $\alpha:[r]\rightarrow [m]$, let 
\begin{equation}
\tl{\alpha}=\{(i,j):j \in Im(\alpha),i>j\}\subset\orderedpairs{m}
\end{equation}
\end{definition}

\begin{definition}
For injective $\beta:[r] \rightarrow [n]$, let 
\begin{multline}
\tr{\beta}=\{(i,j):i \in Im(\beta),j<i,\\
j\notin\{\beta(1),\dots,\beta\left(\beta^{-1}(i)-1\right)\}\} \subset\orderedpairs{n}
\end{multline}
\end{definition}

Note that these sets are transitive and behave naturally under extensions: for $\alpha'$ an extension of $\alpha$, $\tl{\alpha} \subset \tl{\alpha'}$, and for $\beta'$ an extension of $\beta$, $\tr{\beta} \subset \tr{\beta'}$. 

Putting everything so far together gives the following matrix decomposition:

\begin{theorem}\label{thm:GaussianElimination}
For all $A \in \F^{m \times n}$ there exist $r \leq \min(m,n)$, increasing $\alpha:[r]\rightarrow[m]$, injective $\beta:[r]\rightarrow[n]$, and matrices $L \in \L{m}{\tl{\alpha}}$, $R\in\L{n}{\tr{\beta}}$ such that
\begin{equation}
A=L \Pi(\alpha,\beta) R
\end{equation}
$(r,\alpha,\beta,L,R)$ can be computed from $A$ in time $O(mnr)$ by Gaussian elimination.\footnote{Assuming $L,R$ are output as lists of off-diagonal non-zero positions; otherwise, outputting the full $L,R$ would take time $O(\max(m^2,n^2))$.} 
\end{theorem}

\begin{proof}
Encode the evolution of Gaussian elimination in a sequence: take $A_0=A$, and for $i\geq 1$ take
\begin{enumerate}
\item $(\alpha(i),\beta(i))$ to be the position of the $i$-th pivot in the left and down order. 
\item $u_i=A_{i-1} \e{n,\beta(i)} - \e{m,\alpha(i)}$ to be the column of $A_{i-1}$ below the $i$-th pivot.
\item $v_i^\top = \e{m,\alpha(i)}^\top A_{i-1} - \e{n,\beta(i)}^\top$ to be the row of $A_{i-1}$ to the left of the $i$-th pivot. 
\item $A_i = \g{m,u_i,\alpha(i)} A_{i-1} \g{n,\beta(i),v_i^\top}$ to be the remaining matrix after row $\alpha(i)$ and column $\beta(i)$ have been cleared. 
\end{enumerate}
Induction shows that for each $i$, rows $1, \dots, \alpha(i)$ and columns $\beta(1),\dots, \beta(i)$ of $A_i$ match the corresponding rows and columns of $\sum_{k=1}^i \e{m,\alpha(k)}\e{n,\beta(k)}^\top$. Then, the sequence must terminate; let $A_r=\Pi(\alpha,\beta)$ be the last term. Then, $A=L\Pi(\alpha,\beta)R$ with
\begin{equation}
L=\g{m,u_1, \alpha(1)} \g{m,u_2, \alpha(2)} \dots \g{m,u_r,\alpha(r)} \in \L{m}{\tl{\alpha}}
\end{equation}
and 
\begin{equation}
R=\g{n,\beta{r},v_r^\top}\g{n,\beta(r-1),v_{r-1}^\top}\dots\g{n,\beta(1),v_1^\top} \in \L{n}{\tr{\beta}}
\end{equation}

Now, consider the complexity of the algorithm. For each $i$, $(\alpha(i),\beta(i),u_i,v_i,A_i)$ can be computed from $(A_{i-1}, \alpha(i-1))$ in time $O(mn)$. Then, the time to reach step $r$ is $O(mnr)$. 

It remains to consider the time to compute $L$ and $R$. It is shown below that the non-zero off-diagonal positions can be computed in time $O(mr^2)$ for $L$ and in time $O(nr^2)$ for $R$. 

Consider the sequence $L_1 = \g{m,u_1,\alpha(1)}$ and $L_i = L_{i-1}\g{m,u_i,\alpha(i)}$ for $i=2,\dots,r$. Let $\tilde{L}_i$ be such that $L_i = \id{m} + \tilde{L}_i$. Then, 
\begin{equation}
\tilde{L}_i = \tilde{L}_{i-1} + u_i \e{m,\alpha(i)}^\top + \tilde{L}_{i-1} u_i \e{m,\alpha(i)}^\top
\end{equation}
Note that $\tilde{L}_{i-1}, \tilde{L}_i$ are supported on some subsets of $\tl{\alpha}$, so they have at most $r$ non-zero columns. Then, $\tilde{L}_i$ can be computed from $(\tilde{L}_{i-1},u_i, \alpha(i))$ in time $O(mr)$, so $\tilde{L}=\tilde{L}_r$ can be computed from $(u_1,\dots,u_r,\alpha(1),\dots,\alpha(r))$ in time $O(mr^2)$. 

Similarly, consider the sequence $R_1=\g{n,\beta(1),v_1^\top}$ and $R_i=\g{n,\beta(i),v_i^\top} R_{i-1}$ for $i=2,\dots,r$. Let $\tilde{R}_i$ be such that $R_i = \id{n} + \tilde{R}_i$. Then, 
\begin{equation}
\tilde{R}_i = \tilde{R}_{i-1} + \e{n,\beta(i)}v_i^\top + \e{n,\beta(i)} v_i^\top \tilde{R}_{i-1}
\end{equation}
Note that $\tilde{R}_{i-1},\tilde{R}_i$ are supported on some subsets of $\tr{\beta}$, so they have at most $r$ non-zero rows. Then, $\tilde{R}_i$ can be computed from $(\tilde{R}_{i-1},\beta(i),v_i^\top)$ in time $O(nr)$, so $\tilde{R} = \tilde{R}_r$ can be computed from $(\beta(1), \dots, \beta(r), v_1^\top, \dots, v_r^\top)$ in time $O(nr^2)$. 
\end{proof}

The given matrix decomposition is in fact unique; thus, it is a canonical form for $m \times n$ matrices. This will now be proved in detail. 
For convenience, say that a quintuple $(r,\alpha,\beta,L,R)$ is $(m,n)$-allowed if $r \leq \min(m,n)$, $\alpha:[r]\rightarrow[m]$ is increasing, $\beta:[r]\rightarrow[n]$ is injective, $L\in\L{m}{\tl{\alpha}}$ and $R\in\L{n}{\tr{\beta}}$. 

\begin{theorem}\label{thm:UniquenessOfDecomposition}
Let $(r,\alpha,\beta,L,R)$, $(r',\alpha',\beta',L',R')$ be $(m,n)$-allowed quintuples such that
\begin{equation}
L\Pi(\alpha,\beta)R=L'\Pi(\alpha',\beta')R'
\end{equation}
Then, 
\begin{equation}
(r,\alpha,\beta,L,R)=(r',\alpha',\beta',L',R')
\end{equation}
\end{theorem} 

\begin{proof}
The first step is to establish uniqueness of the incomplete permutation matrix. In \cite[Lemma 14]{bravyi2021hadamardfree}, an argument for the case of full permutation matrices is given; unfortunately, it is not clear how to extend this approach to the case of incomplete permutation matrices. On the other hand, a rank-of-minors argument is given in \cite{strang2015algebra}, also for the case of full permutation matrices, and this proof extends without problem to the case of incomplete permutation matrices. 

\begin{lemma}\label{lemma:UniquenessOfIncompletePermutation}
Let $\Pi(\alpha,\beta)$, $\Pi(\alpha',\beta')$ be two $m \times n$ incomplete permutation matrices. If there exist lower triangular $L \in \L{m}{\orderedpairs{m}}$, $R \in \L{n}{\orderedpairs{n}}$ such that 
\begin{equation}
L \Pi(\alpha,\beta)=\Pi(\alpha',\beta')R
\end{equation}
then $\Pi(\alpha,\beta)=\Pi(\alpha',\beta')$. 
\end{lemma}

\begin{proof}
Take any $k\in[m],l\in[n]$, and divide the matrices into blocks with $(k,m-k)$ and $(l,n-l)$ rows/columns:
\begin{multline}
\twobytwo{L_{11}}{0}{L_{21}}{L_{22}} \twobytwo{\Pi(\alpha,\beta)_{11}}{\Pi(\alpha,\beta)_{12}}{\Pi(\alpha,\beta)_{21}}{\Pi(\alpha,\beta)_{22}} \\
= \twobytwo{\Pi(\alpha',\beta')_{11}}{\Pi(\alpha',\beta')_{12}}{\Pi(\alpha',\beta')_{21}}{\Pi(\alpha',\beta')_{22}} \twobytwo{R_{11}}{0}{R_{21}}{R_{22}}
\end{multline}
Then, the upper right blocks $\Pi(\alpha,\beta)_{12},\Pi(\alpha',\beta')_{12}$ have the same rank. This holds for any $k,l$, so $\Pi(\alpha,\beta)=\Pi(\alpha',\beta')$. 
\end{proof}

The next step is to establish uniqueness of the right factor. For the case of full permutation matrices, both \cite{strang2015algebra,bravyi2021hadamardfree} argue by considering whether conjugating a lower triangular matrix by a permutation leads to another lower triangular matrix. Unfortunately, this approach breaks down for incomplete permutation matrices, and a different approach is needed. 

\begin{lemma}\label{lemma:UniquenessOfRightFactor}
Let $\beta:[r]\rightarrow[n]$ be injective, and let $R$ be in $\L{n}{\tr{\beta}}$. If there exist $m$, lower triangular $L \in \L{m}{\orderedpairs{m}}$, and increasing $\alpha:[r]\rightarrow[m]$ such that $L\Pi(\alpha,\beta)=\Pi(\alpha,\beta)R$, then $R=I$. 
\end{lemma}

\begin{proof}
Take any $i,j\in[n]$. The goal is to show $\e{n,i}^{\top} R \e{n,j} = \e{n,i}^{\top}\e{n,j}$. Consider cases. 

\emph{Case 1: $i \leq j$.} Then, $\e{n,i}^{\top} R \e{n,j} = \e{n,i}^{\top}\e{n,j}$ because $R$ is lower triangular. 

\emph{Case 2: $i>j$, $i \notin Im(\beta)$.} Then, $\e{n,i}^{\top} R\e{n,j}=0$ because $R \in \L{n}{\tr{\beta}}$. 

\emph{Case 3: $i>j$, $\exists k:i=\beta(k)$, $j \notin Im(\beta)$.} Then,
\begin{equation}
\e{n,i}^{\top} R \e{n,j} = \e{m,\alpha(k)}^{\top} \Pi(\alpha,\beta) R \e{n,j} 
 = \e{m,\alpha(k)}^{\top} L \Pi(\alpha,\beta) \e{n,j} = 0
\end{equation}

\emph{Case 4: $i>j$, $\exists k,l: i=\beta(k),j=\beta(l),l<k$.} Then, $\e{n,i}^{\top} R \e{n,j}=0$ because $R \in \L{n}{\tr{\beta}}$. 

\emph{Case 5: $i>j$, $\exists k,l: i=\beta(k),j=\beta(l),k<l$.} Then,
\begin{multline}
\e{n,i}^{\top} R \e{n,j} = \e{m,\alpha(k)}^{\top} \Pi(\alpha,\beta) R \e{n,j}\\
= \e{m,\alpha(k)}^{\top} L \Pi(\alpha,\beta)\e{n,j} = \e{m,\alpha(k)}^{\top} L \e{m,\alpha(l)} = 0
\end{multline}
because $\alpha$ is increasing and $L$ is lower triangular. 
\end{proof}

Finally, note that since incomplete permutation matrices are not necessarily invertible, uniqueness of the left factor does not follow automatically from uniqueness of the incomplete permutation matrix and of the right factor. Fortunately, the following Lemma holds:

\begin{lemma}\label{lemma:UniquenessOfLeftFactor}
Let $\alpha:[r]\rightarrow[m]$ be increasing, and let $L$ be in $\L{m}{\tl{\alpha}}$. If there exist $n$ and injective $\beta:[r]\rightarrow[n]$ such that $L\Pi(\alpha,\beta)=\Pi(\alpha,\beta)$, then $L=I$. 
\end{lemma}

\begin{proof}
Take any $i \in [m]$. If $i\notin Im(\alpha)$, then $L\e{m,i}=\e{m,i}$ by the definition of $\L{m}{\tl{\alpha}}$. If $i \in Im(\alpha)$, then
\begin{equation}
L\e{m,i}=L\Pi(\alpha,\beta)\e{n,\beta(\alpha^{-1}(i))}=\Pi(\alpha,\beta)\e{n,\beta(\alpha^{-1}(i))}=\e{m,i}
\end{equation}  
Thus, $L\e{m,i}=\e{m,i}$ holds for all $i\in[m]$.  
\end{proof}

Lemmas \ref{lemma:UniquenessOfIncompletePermutation}, \ref{lemma:UniquenessOfRightFactor}, \ref{lemma:UniquenessOfLeftFactor} imply Theorem \ref{thm:UniquenessOfDecomposition}. 
\end{proof}

\subsection{Canonical form for $m \times 2n$ stabilizer parity check matrices}\label{sec:canonical_form_for_stabilizer_pcm}

Now, an application of the groups $\B{2n}{T}$ will be described. Recall from section \ref{sec:CanonicalFormForUnrestrictedMatrices} that Gaussian elimination produces a canonical form for unrestricted matrices in $\F^{m \times n}$. Now, consider $m \times 2n$ stabilizer parity check matrices. This section shows that Gaussian elimination modified to use symplectic Gaussian moves on the right produces a canonical form for such matrices, and the right factor belongs to a group $\B{2n}{T}$ for $T$ a suitable function of the pivot positions. 

First, some definitions:

\begin{definition}
$A \in \F^{m \times 2n}$ is a stabilizer parity check matrix if $A \revdiag{2n} A^{\top}=0$. 
\end{definition}

When Gaussian elimination is performed on a stabilizer parity check matrix with symplectic moves on the right, the column indices of the pivots have special structure:

\begin{definition}
Let $\qubit{n}:[2n]\rightarrow[n]$ be the function
\begin{equation}
\qubit{n}(i) = \min(i,2n+1-i)
\end{equation}
that maps each $i\in [2n]$ to the qubit on which Pauli $\q{2n,i}$ acts.
\end{definition}

\begin{definition}
A function $\beta:[r]\rightarrow[2n]$ will be called qubit-injective if $\qubit{n} \circ\beta:[r]\rightarrow[n]$ is injective.  
\end{definition}

As in the classical case, there is a set of off-diagonal positions that are associated to a collection of column indices:

\begin{definition}
To qubit-injective $\beta:[r]\rightarrow[2n]$ associate the sets
\begin{multline}
\tm{\beta}=\{(i,j):i\in Im(\beta),j<i,\\
\qubit{n}(j) \notin \{\qubit{n}(\beta(1)),\dots,\qubit{n}(\beta(\beta^{-1}(i)-1)) \}\}
\end{multline}
\begin{equation}
\tmr{\beta}=\left\{(i,j):(2n+1-j,2n+1-i)\in\tm{\beta}\right\}
\end{equation}
\begin{equation}
\ttcr{\beta}=\tm{\beta}\cup\tmr{\beta}
\end{equation}
\end{definition}

$\tm{\beta}$ is the set of index pairs $(i,j)$ such that $\s{2n,i,j}$ can potentially be used during the elimination algorithm described below when the column indices of the pivots are given by $\beta$. $\tmr{\beta}$ is the elementwise reversal of $\tm{\beta}$; since $\s{2n,i,j}=\s{2n,2n+1-j,2n+1-i}$, it is natural to consider also this set. $\ttcr{\beta}$ is the transitive, closed under reversal set that defines the associated group $\B{2n}{\ttcr{\beta}}$. Relevant properties of $\ttcr{\beta}$ are shown below:

\begin{lemma}
If $\beta:[r]\rightarrow[2n]$ is qubit-injective, then $\ttcr{\beta}$ is transitive and closed under reversal. 

If $\beta'$ is qubit-injective and extends $\beta$, then $\ttcr{\beta}\subset\ttcr{\beta'}$. 
\end{lemma}

\begin{proof}
$\ttcr{\beta}$ is closed under reversal by construction. 

Now, take any $(i,j),(j,k)\in\ttcr{\beta}$ and consider cases:

\emph{Case 1: $(i,j),(j,k)\in\tm{\beta}$.} Then, $i,j\in Im(\beta)$ with $\beta^{-1}(j)>\beta^{-1}(i)$. Moreover, $i>j>k$ and
\begin{multline}
\qubit{n}(k)\notin \{\qubit{n}(\beta(1)), \dots, \qubit{n}(\beta(\beta^{-1}(j)-1))\} \\
\supset \{\qubit{n}(\beta(1)),\dots,\qubit{n}(\beta(\beta^{-1}(i)-1))\}
\end{multline}
Then, $(i,k) \in \tm{\beta}$. 

\emph{Case 2: $(i,j),(j,k)\in\tmr{\beta}$.} Then, $(2n+1-k,2n+1-j),(2n+1-j,2n+1-i)\in\tm{\beta}$, so $(2n+1-k,2n+1-i)\in\tm{\beta}$, so $(i,k)\in\tmr{\beta}$. 

\emph{Case 3: $(i,j)\in\tm{\beta},(j,k)\in\tmr{\beta}$.} Then $i>j>k$ and $i,2n+1-k \in Im(\beta)$. Consider subcases: 
\begin{enumerate}
\item If $\beta^{-1}(2n+1-k) \geq \beta^{-1}(i)$, then $(i,k) \in \tm{\beta}$.
\item If $\beta^{-1}(2n+1-k) < \beta^{-1}(i)$, then $(2n+1-k,2n+1-i) \in \tm{\beta}$, so $(i,k)\in \tmr{\beta}$. 
\end{enumerate}

\emph{Case 4: $(i,j)\in\tmr{\beta},(j,k)\in\tm{\beta}$.} Then, $j,2n+1-j \in Im(\beta)$. This is a contradiction, because $\beta$ is qubit injective. This case does not occur. 

Combining the cases implies that $\ttcr{\beta}$ is transitive. 

Finally, if $\beta'$ extends $\beta$, then $\tm{\beta'}\supset\tm{\beta}$, so $\ttcr{\beta'}\supset\ttcr{\beta}$. 
\end{proof}

Now, consider Gaussian elimination using symplectic moves on the right. The typical step of this algorithm is:
\begin{lemma}\label{lemma:RightSymplecticEliminationStep}
Let $\alpha:[r]\rightarrow[m]$ be increasing and $\beta:[r]\rightarrow[2n]$ be qubit-injective. Let $A \in \F^{m \times 2n}$ be a stabilizer parity check matrix that is $(\alpha,\beta)$ partially reduced, but not fully reduced. Let $(\alpha',\beta')$ extend $(\alpha,\beta)$ so that $(\alpha'(r+1),\beta'(r+1))$ is the position of the next pivot in the left and down order. Let $u \in \F^m$, $v\in\F^{2n}$ be such that 
\begin{align}
A \e{2n,\beta'(r+1)} &= \e{m,\alpha'(r+1)} + u \\
\e{m,\alpha'(r+1)}^{\top} A &= \e{2n,\beta'(r+1)}^{\top} + v^{\top}
\end{align}
and let
\begin{equation}
A' = \g{m,u,\alpha'(r+1)} A \s{2n,\beta'(r+1),v^{\top}}
\end{equation}
Then, 
\begin{enumerate}
\item $\g{m,u,\alpha'(r+1)} \in \L{m}{\tl{\alpha'}}$.
\item $\beta'$ is qubit injective.
\item $\s{2n,\beta'(r+1),v^{\top}} \in \B{2n}{\ttcr{\beta'}}$. 
\item $A'$ is a stabilizer parity check matrix and is $(\alpha',\beta')$ partially reduced. 
\item $(\alpha',\beta',u,v,A')$ can be computed from $(A,\alpha,\beta)$ in time $O(mn)$. 
\end{enumerate}
\end{lemma}

\begin{proof}
\emph{Part 1:} 
Since $A$ is $(\alpha,\beta)$ reduced, $u$ is supported on a subset of $\{j:j>\alpha'(r+1)\}$.  

\emph{Parts 2 and 3:} 
Since $A$ is $(\alpha,\beta)$ reduced, and since $A \revdiag{2n}A^{\top}=0$, the following holds for the columns of $A$:
\begin{equation}
\forall i \in [r]: \;\; A \e{2n,\beta(i)}= \e{m,\alpha(i)} \;\; \wedge \;\; A\e{2n,2n+1-\beta(i)}=0
\end{equation}
Therefore, $\qubit{n}(\beta'(r+1)) \notin \{\qubit{n}(\beta(1)),\dots\qubit{n}(\beta(r)\}$, and $v$ is supported on a subset of $\{i:i<\beta'(r+1),\qubit{n}(i)\notin\{\qubit{n}(\beta(1)), \dots, \qubit{n}(\beta(r))\}\}$. 

\emph{Part 4:} First,  
\begin{multline}
A' \revdiag{2n}(A')^{\top} 
=\g{m,u,\alpha'(r+1)}A\s{2n,\beta'(r+1),v^{\top}}\revdiag{2n}
\s{2n,\beta'(r+1),v^{\top}}^{\top}A^{\top}\g{m,u,\alpha'(r+1)}^{\top}\\
=\g{m,u,\alpha'(r+1)}A\revdiag{2n}A^{\top}\g{m,u,\alpha'(r+1)}^{\top}=0
\end{multline}
Next, check the rows: for $i\leq \alpha'(r+1)$,
\begin{equation}
\e{m,i}^{\top} A' = \e{m,i}^{\top}\g{m,u,\alpha'(r+1)}A\s{2n,\beta'(r+1),v^{\top}} 
=\e{m,i}^{\top}A\s{2n,\beta'(r+1),v^{\top}}
\end{equation}
Consider cases:
\begin{enumerate}
\item If $ i \notin Im(\alpha') $ then $ \e{m,i}^{\top}A\s{2n,\beta'(r+1),v^{\top}} = 0$
\item If $ i=\alpha(k), k \in [r] $ then
\begin{multline*}
\e{m,i}^{\top}A\s{2n,\beta'(r+1),v^{\top}} = \e{m,\alpha(k)}^{\top}A\s{2n,\beta'(r+1),v^{\top}} \\
= \e{2n,\beta(k)}^{\top} \s{2n,\beta'(r+1),v^{\top}} = \e{2n,\beta(k)}^{\top}
\end{multline*}
\item If $ i = \alpha'(r+1) $ then
\begin{equation*}
\e{m,i}^{\top}A\s{2n,\beta'(r+1),v^{\top}} 
= (\e{2n,\beta'(r+1)}^{\top}+v^{\top})\s{2n,\beta'(r+1),v^{\top}} = \e{2n,\beta'(r+1)}^{\top}
\end{equation*}
\end{enumerate}
Next, check the columns: for $i \in [r+1]$
\begin{multline}
A' \e{2n,\beta'(i)}=\g{m,u,\alpha'(r+1)}A\s{2n,\beta'(r+1),v^{\top}}\e{2n,\beta'(i)}\\
=\g{m,u,\alpha'(r+1)}A\e{2n,\beta'(i)}
\end{multline}
Consider cases:
\begin{enumerate}
\item If $ i=r+1$ then 
\begin{equation*}
\g{m,u,\alpha'(r+1)}A\e{2n,\beta'(i)} 
= \g{m,u,\alpha'(r+1)}(\e{m,\alpha'(r+1)}+u) = \e{m,\alpha'(r+1)} 
\end{equation*}
\item If $ i \leq r $ then 
\begin{equation*}
\g{m,u,\alpha'(r+1)}A\e{2n,\beta'(i)} 
= \g{m,u,\alpha'(r+1)}\e{m,\alpha(i)}=\e{m,\alpha(i)} 
\end{equation*}
\end{enumerate}

\emph{Part 5:} Computation of $(\alpha',\beta',u,v)$ can clearly be done in time $O(mn)$. For the computation of $A'$, use the sparsity of $\g{m,u,\alpha'(r+1)}, \s{2n,\beta'(r+1),v^{\top}}$. For any $x \in \F^m$, $\g{m,u,\alpha'(r+1)}x$ can be computed in time $O(m)$; since $A$ has $2n$ columns, this gives time $O(mn)$ to compute $\g{m,u,\alpha'(r+1)}A$. Similarly, for any $y \in \F^{2n}$, $y^{\top}\s{2n,\beta'(r+1),v^{\top}}$ can be computed in time $O(n)$; for $m$ rows, this gives time $O(mn)$ to compute $\g{m,u,\alpha'(r+1)}A\s{2n,\beta'(r+1),v^{\top}}$. 
\end{proof}

Repeated application of the elimination step in Lemma \ref{lemma:RightSymplecticEliminationStep} produces a matrix decomposition for stabilizer parity check matrices that is in fact a canonical form. To prove uniqueness of the decomposition, a modification of Lemma \ref{lemma:UniquenessOfRightFactor} is needed:
\begin{lemma}\label{lemma:UniquenessOfSymplecticRightFactor}
Let $\beta:[r]\rightarrow[2n]$ be qubit-injective and let $R \in \B{2n}{\ttcr{\beta}}$. If there exist lower triangular $L\in\L{m}{\orderedpairs{m}}$ and increasing $\alpha:[r]\rightarrow [m]$ such that $L\Pi(\alpha,\beta)=\Pi(\alpha,\beta)R$, then $R=I$.  
\end{lemma}

\begin{proof}
Take any $(i,j) \in \ttcr{\beta}$. The goal is to show $\e{2n,i}^{\top} R\e{2n,j}=0$. Consider cases. 

\emph{Case 1: $(i,j)\in\tm{\beta}$.} Then, $i=\beta(k)$ for some $k\in[r]$ and
\begin{equation}
\e{2n,i}^{\top} R \e{2n,j}=\e{m,\alpha(k)}^{\top}\Pi(\alpha,\beta)R\e{2n,j}
=\e{m,\alpha(k)}^{\top} L \Pi(\alpha,\beta)\e{2n,j}
\end{equation}
Consider subcases:
\begin{enumerate}
\item If $j \notin Im(\beta)$, then $\Pi(\alpha,\beta)\e{2n,j}=0$.
\item If $j=\beta(l)$ for some $l\in[r]$, then $l>k$ and $\e{m,\alpha(k)}L\e{m,\alpha(l)}=0$ because $\alpha$ is increasing and $L$ is lower triangular.
\end{enumerate} 

\emph{Case 2: $(i,j) \in \tmr{\beta}$.} Then, $(2n+1-j,2n+1-i) \in \tm{\beta}$ and $2n+1-j=\beta(k)$ for some $k \in [r]$. Apply the identity $R^{\top} = \revdiag{2n} R^{-1} \revdiag{2n}$, which holds for any element of $\symplecticgroup{2n}$:
\begin{multline}
\e{2n,i}^{\top} R \e{2n,j} = \e{2n,j}^{\top} R^{\top} \e{2n,i} \\
= \e{2n,j}^{\top} \revdiag{2n} R^{-1} \revdiag{2n} \e{2n,i} \\
 = \e{2n,2n+1-j}^{\top} R^{-1} \e{2n,2n+1-i} \\ = \e{m,\alpha(k)}\Pi(\alpha,\beta) R^{-1} \e{2n,2n+1-i} \\
= \e{m,\alpha(k)} L^{-1} \Pi(\alpha,\beta) \e{2n,2n+1-i}
\end{multline} 
Now, proceed as in case 1 to conclude that this is zero. 
\end{proof}

Now, all the ingredients for a canonical form for stabilizer parity check matrices are in place. To state the theorem, it is convenient to introduce the following terminology: 
\begin{definition}
Say that $(r,\alpha,\beta,L,R)$ is $(m,2n)$ stabilizer allowed if $r \leq \min(m,n)$, $\alpha:[r]\rightarrow[m]$ is increasing, $\beta:[r]\rightarrow[2n]$ is qubit-injective, $L\in \L{m}{\tl{\alpha}}$, $R \in \B{2n}{\ttcr{\beta}}$. 
\end{definition}

\begin{theorem}\label{thm:canonical_form_for_stabilizer_pcm}
\begin{enumerate}
\item For every stabilizer parity check matrix $A \in \F^{m \times 2n}$, Gaussian elimination with symplectic moves on the right produces $(m,2n)$ stabilizer allowed $(r,\alpha,\beta,L,R)$ such that 
\begin{equation}
A=L\Pi(\alpha,\beta)R
\end{equation}
Moreover, the algorithm runs in time $O(mnr)$ (assuming $L,R$ are output as lists of off-diagonal non-zero positions; otherwise, outputting the full $L,R$ would take time $O(\max(m^2,n^2))$).  
\item If $(r,\alpha,\beta, L,R)$ and $(r',\alpha',\beta',L',R')$ are $(m,2n)$ stabilizer allowed and if 
\begin{equation}
L\Pi(\alpha,\beta)R=L'\Pi(\alpha',\beta')R'
\end{equation}
then
\begin{equation}
(r,\alpha,\beta,L,R)=(r',\alpha',\beta',L',R')
\end{equation}
\end{enumerate}
\end{theorem}

\begin{proof}
\emph{Part 1} follows by repeated application of Lemma \ref{lemma:RightSymplecticEliminationStep}. The elimination step in Lemma \ref{lemma:RightSymplecticEliminationStep} needs to be performed $r$ times, so the total time is $O(mnr)$. It remains to consider the time to compute $L$ and $R$. It was shown in the proof of theorem \ref{thm:GaussianElimination} that the non-zero off-diagonal entries of $L \in \L{n}{\tl{\alpha}}$ can be computed in time $O(mr^2)$. Next, consider the computation of 
\begin{equation}
R= \s{2n,\beta(r),v_r^\top}\s{2n,\beta(r-1),v_{r-1}^\top}\dots\s{2n,\beta(1),v_1^\top}
\in\B{2n}{\ttcr{\beta}}
\end{equation}
where $v_1^\top, \dots, v_r^\top$ is the sequence of vectors that arises from the $r$ applications of Lemma \ref{lemma:RightSymplecticEliminationStep}. Consider the sequence $R_1=\s{2n,\beta(1),v_1^\top}$ and $R_i = \s{2n,\beta(i),v_i^\top} R_{i-1}$ for $i=2,\dots,r$. Let $\tilde{R}_i$ be such that $R_i = \id{2n}+\tilde{R}_i$ and let $\tilde{S}_i$ be such that $\s{2n,\beta(i),v_i^\top}=\id{2n}+\tilde{S}_i$. Then, 
\begin{equation}
\tilde{R}_i = \tilde{R}_{i-1} + \tilde{S}_i + \tilde{S}_i \tilde{R}_{i-1}
\end{equation}
Note that $\tilde{R}_{i-1}, \tilde{R}_i$ are supported on some subsets of $\ttcr{\beta}$, which in turn is a subset of $\orderedpairs{2n} \cap ((Im(\beta) \times [2n]) \cup ([2n] \times (2n+1-Im(\beta))))$. Note also that $\tilde{S}_i$ has non-zero entries only in row $\beta(i)$ and column $2n+1-\beta(i)$. Therefore,

\textbf{Claim:} $\tilde{R}_i$ can be computed from $(\tilde{R}_{i-1},\tilde{S}_i)$ in time $O(nr)$. 

\textbf{Proof of claim:} Partition $\ttcr{\beta}$ into three sets: 
\begin{align}
T_1 &= \curlybr{(k,l)\in\ttcr{\beta}:k\neq\beta(i)}\\
T_2 &= \curlybr{(k,l)\in\ttcr{\beta}:k=\beta(i), l \notin (2n+1-Im(\beta))}\\
T_3 &= \curlybr{(k,l)\in\ttcr{\beta}:k=\beta(i), l \in (2n+1-Im(\beta)}
\end{align} 
and note that $|T_1|=O(nr)$, $|T_2| = O(n)$, $|T_3|=O(r)$. For each $(k,l) \in T_1$, 
\begin{multline}
\e{2n,k}^\top \tilde{R}_i \e{2n,l} = \e{2n,k}^\top \tilde{R}_{i-1} \e{2n,l} + \e{2n,k}^\top \tilde{S}_i \e{2n,l} \\
+ \e{2n,k}^\top \tilde{S}_i \e{2n,2n+1-\beta(i)}\e{2n,2n+1-\beta(i)}^\top \tilde{R}_{i-1} \e{2n,l}
\end{multline}
can be computed in time $O(1)$. For each $(k,l) \in T_2$, 
\begin{multline}
\e{2n,k}^\top \tilde{R}_i \e{2n,l} = \e{2n,k}^\top \tilde{R}_{i-1} \e{2n,l} + \e{2n,k}^\top \tilde{S}_i \e{2n,l} \\
+ \sum_{j \in Im(\beta)}\e{2n,k}^\top \tilde{S}_i  \e{2n,j}\e{2n,j}^\top \tilde{R}_{i-1} \e{2n,l}
\end{multline}
can be computed in time $O(r)$. For each $(k,l) \in T_3$, 
\begin{multline}
\e{2n,k}^\top \tilde{R}_i \e{2n,l} = \e{2n,k}^\top \tilde{R}_{i-1} \e{2n,l} + \e{2n,k}^\top \tilde{S}_i \e{2n,l} \\
+ \sum_{j \in [2n]}\e{2n,k}^\top \tilde{S}_i  \e{2n,j}\e{2n,j}^\top \tilde{R}_{i-1} \e{2n,l}
\end{multline}
can be computed in time $O(n)$. This completes the proof of the claim. 

Using the claim, deduce that the non-zero off-diagonal positions of $R$ can be computed from $(\beta(1),\dots,\beta(r),v_1^\top, \dots, v_r^\top)$ in time $O(nr^2)$. This completes the proof of Part 1. 

\emph{Part 2:} First, Lemma \ref{lemma:UniquenessOfIncompletePermutation} implies $\Pi(\alpha,\beta)=\Pi(\alpha',\beta')$. Next, Lemma \ref{lemma:UniquenessOfSymplecticRightFactor} implies $R=R'$. Finally, Lemma \ref{lemma:UniquenessOfLeftFactor} implies $L=L'$. 
\end{proof}

\emph{Example:} Consider the $[[5,1,3]]$ stabilizer code with generators $\x{5,1}\z{5,2}\z{5,3}\x{5,4}$, $\x{5,2}\z{5,3}\z{5,4}\x{5,5}$, $\x{5,1}\x{5,3}\z{5,4}\z{5,5}$ and $\z{5,1}\x{5,2}\x{5,4}\z{5,5}$. Consider also the product of the first two generators $\x{5,1}\z{5,2}\x{5,2}\x{5,4}\z{5,4}\x{5,5}$. This gives the stabilizer parity check matrix 
\begin{equation}
A=
\begin{pmatrix}
1&0&0&1&0&0&0&1&1&0 \\
0&1&0&0&1&0&1&1&0&0 \\
1&1&0&1&1&0&1&0&1&0 \\
1&0&1&0&0&1&1&0&0&0 \\
0&1&0&1&0&1&0&0&0&1
\end{pmatrix}
\end{equation}
where the third row is the sum of the first two. 

The first pivot in the left and down order is in the first row and ninth column, so $\alpha(1)=1$, $\beta(1)=9$. To clear the ninth column, take 
\begin{equation}
u=
\begin{pmatrix}
0\\0\\1\\0\\0
\end{pmatrix}
\; \text{ and } \; \g{5, u, 1} =
\begin{pmatrix}
1&0&0&0&0\\
0&1&0&0&0\\
1&0&1&0&0\\
0&0&0&1&0\\
0&0&0&0&1
\end{pmatrix}
\end{equation}
At this stage, $\tl{\alpha}=\{(2,1),(3,1),(4,1),(5,1)\}$ and $\g{5,u,1}$ has a non-zero off-diagonal position only at $(3,1)$. Left multiplication by $\g{5,u,1}$ gives
\begin{equation}
\g{5,u,1}A=
\begin{pmatrix}
1&0&0&1&0&0&0&1&1&0 \\
0&1&0&0&1&0&1&1&0&0 \\
0&1&0&0&1&0&1&1&0&0 \\
1&0&1&0&0&1&1&0&0&0 \\
0&1&0&1&0&1&0&0&0&1
\end{pmatrix}
\end{equation}
To clear the first row, take 
\begin{equation}
v^\top = \begin{pmatrix}1&0&0&1&0&0&0&1&0&0\end{pmatrix}
\end{equation}
and
\begin{equation}
\s{10,9,v^\top}=
\begin{pmatrix}
1&0&0&0&0&0&0&0&0&0\\
0&1&0&0&0&0&0&0&0&0\\
0&1&1&0&0&0&0&0&0&0\\
0&0&0&1&0&0&0&0&0&0\\
0&0&0&0&1&0&0&0&0&0\\
0&0&0&0&0&1&0&0&0&0\\
0&1&0&0&0&0&1&0&0&0\\
0&0&0&0&0&0&0&1&0&0\\
1&0&0&1&0&0&0&1&1&0\\
0&1&0&0&0&0&0&0&0&1
\end{pmatrix}
\end{equation}
At this stage, $\tm{\beta}=\{(9,j):j<9\}$, $\tmr{\beta}=\{(i,2):i>2\}$, $\ttcr{\beta}=\tm{\beta}\cup\tmr{\beta}$ and the non-zero off-diagonal positions of $\s{10,9,v^\top}$ are only in the second column and the ninth row. Right multiplication by $\s{10,9,v^\top}$ gives
\begin{equation}
\g{5,u,1}A\s{10,9,v^\top}=
\begin{pmatrix}
0&0&0&0&0&0&0&0&1&0 \\
0&0&0&0&1&0&1&1&0&0 \\
0&0&0&0&1&0&1&1&0&0 \\
1&0&1&0&0&1&1&0&0&0 \\
0&0&0&1&0&1&0&0&0&1
\end{pmatrix}
\end{equation}
At this point, the first row and ninth column contain only the first pivot. The second column, which corresponds to the same qubit as the ninth column, is entirely cleared. 

The second pivot in the left and down order is in the second row and eigth column. Extend $\alpha, \beta$ to $\alpha', \beta'$ with $\alpha'(2)=2$, $\beta'(2)=8$. To clear the eight column, take 
\begin{equation}
u'=
\begin{pmatrix}
0\\0\\1\\0\\0
\end{pmatrix}
\; \text{ and } \; \g{5, u', 2} =
\begin{pmatrix}
1&0&0&0&0\\
0&1&0&0&0\\
0&1&1&0&0\\
0&0&0&1&0\\
0&0&0&0&1
\end{pmatrix}
\end{equation}
At this stage, $\tl{\alpha'}=\{(2,1),(3,1),(4,1),(5,1),(3,2),(4,2),(5,2)\}$ and the non-zero off-diagonal positions of $\g{5,u',2}$ belong to this set. Left multiplication by $\g{5,u',2}$ gives
\begin{equation}
\g{5,u',2}\g{5,u,1}A\s{10,9,v^\top}
=
\begin{pmatrix}
0&0&0&0&0&0&0&0&1&0 \\
0&0&0&0&1&0&1&1&0&0 \\
0&0&0&0&0&0&0&0&0&0 \\
1&0&1&0&0&1&1&0&0&0 \\
0&0&0&1&0&1&0&0&0&1
\end{pmatrix}
\end{equation}
To clear the second row, take
\begin{equation}
{v'}^\top = \begin{pmatrix}0&0&0&0&1&0&1&0&0&0\end{pmatrix}
\end{equation}
and
\begin{equation}
\s{10,8,{v'}^\top}=
\begin{pmatrix}
1&0&0&0&0&0&0&0&0&0\\
0&1&0&0&0&0&0&0&0&0\\
0&0&1&0&0&0&0&0&0&0\\
0&0&1&1&0&0&0&0&0&0\\
0&0&0&0&1&0&0&0&0&0\\
0&0&1&0&0&1&0&0&0&0\\
0&0&0&0&0&0&1&0&0&0\\
0&0&0&0&1&0&1&1&0&0\\
0&0&0&0&0&0&0&0&1&0\\
0&0&0&0&0&0&0&0&0&1
\end{pmatrix}
\end{equation}
At this stage, 
\begin{align}
\tm{\beta'}&=\{(9,j):j<9\}\cup\{(8,j):j<8,j \neq 2\} \\
\tmr{\beta'}&=\{(i,2):i>2\}\cup\{(i,3):i>3, i \neq 9\}
\end{align}
and $\ttcr{\beta'}=\tm{\beta'}\cup\tmr{\beta'}$. The non-zero off-diagonal positions of $\s{10,8,{v'}^\top}$ belong to this set. Right multiplication by $\s{10,8,{v'}^\top}$ gives
\begin{equation}
\g{5,u',2}\g{5,u,1}A\s{10,9,v^\top}\s{10,8,{v'}^\top}
=
\begin{pmatrix}
0&0&0&0&0&0&0&0&1&0 \\
0&0&0&0&0&0&0&1&0&0 \\
0&0&0&0&0&0&0&0&0&0 \\
1&0&0&0&0&1&1&0&0&0 \\
0&0&0&1&0&1&0&0&0&1
\end{pmatrix}
\end{equation}
At this point, the second row and eight column contain only the second pivot. The third column, which corresponds to the same qubit as the eight, is completely clear. 

The third pivot in the left and down order is in the fourth row and seventh column. Extend $\alpha',\beta'$ to $\alpha'',\beta''$ where $\alpha''(3)=4$, $\beta''(3)=7$. The seventh column is already clear except for the pivot. To clear the fourth row, take
\begin{equation}
{v''}^\top = \begin{pmatrix}1&0&0&0&0&1&0&0&0&0\end{pmatrix}
\end{equation}
and
\begin{equation}
\s{10,7,{v''}^\top}=
\begin{pmatrix}
1&0&0&0&0&0&0&0&0&0\\
0&1&0&0&0&0&0&0&0&0\\
0&0&1&0&0&0&0&0&0&0\\
0&0&0&1&0&0&0&0&0&0\\
0&0&0&1&1&0&0&0&0&0\\
0&0&0&0&0&1&0&0&0&0\\
1&0&0&0&0&1&1&0&0&0\\
0&0&0&0&0&0&0&1&0&0\\
0&0&0&0&0&0&0&0&1&0\\
0&0&0&1&0&0&0&0&0&1
\end{pmatrix}
\end{equation}
At this stage, 
\begin{multline}
\tm{\beta''}=\{(9,j):j<9\}\cup\{(8,j):j<8,j \neq 2\}\\
\cup\{(7,j):j<7,j \neq 2,3\}
\end{multline}
\begin{multline}
\tmr{\beta''}=\{(i,2):i>2\}\cup\{(i,3):i>3, i \neq 9\}\\
\cup\{(i,4):i>4,i\neq 8,9\}
\end{multline}
and $\ttcr{\beta''}=\tm{\beta''}\cup\tmr{\beta''}$. The non-zero off-diagonal positions of $\s{10,7,{v''}^\top}$ belong to this set. Right multiplication by $\s{10,7,{v''}^\top}$ gives
\begin{multline}
\g{5,u',2}\g{5,u,1}A\s{10,9,v^\top}\s{10,8,{v'}^\top}\s{10,7,{v''}^\top}\\
=
\begin{pmatrix}
0&0&0&0&0&0&0&0&1&0 \\
0&0&0&0&0&0&0&1&0&0 \\
0&0&0&0&0&0&0&0&0&0 \\
0&0&0&0&0&0&1&0&0&0 \\
0&0&0&0&0&1&0&0&0&1
\end{pmatrix}
\end{multline}
At this point, the fourth row and seventh column contain only the third pivot. The fourth column, which corresponds to the same qubit as the seventh, is completely clear. 

The fourth pivot in the left and down order is in the fifth row and tenth column. Extend $\alpha'',\beta''$ to $\alpha''',\beta'''$ where $\alpha'''(4)=5$, $\beta'''(4)=10$. The tenth column is already clear except for the pivot. To clear the fifth row, take
\begin{equation}
{v'''}^\top = \begin{pmatrix}0&0&0&0&0&1&0&0&0&0\end{pmatrix}
\end{equation}
and
\begin{equation}
\s{10,10,{v'''}^\top}=
\begin{pmatrix}
1&0&0&0&0&0&0&0&0&0\\
0&1&0&0&0&0&0&0&0&0\\
0&0&1&0&0&0&0&0&0&0\\
0&0&0&1&0&0&0&0&0&0\\
1&0&0&0&1&0&0&0&0&0\\
0&0&0&0&0&1&0&0&0&0\\
0&0&0&0&0&0&1&0&0&0\\
0&0&0&0&0&0&0&1&0&0\\
0&0&0&0&0&0&0&0&1&0\\
0&0&0&0&0&1&0&0&0&1
\end{pmatrix}
\end{equation}
At this stage, 
\begin{multline}
\tm{\beta'''}=\{(9,j):j<9\}\cup\{(8,j):j<8,j \neq 2\} \\
\cup\{(7,j):j<7,j \neq 2,3\}\\
\cup\{(10,j):j<10, j \neq 2,3,4,7,8,9\}
\end{multline}
\begin{multline}
\tmr{\beta'''}=\{(i,2):i>2\}\cup\{(i,3):i>3, i \neq 9\}\\
\cup\{(i,4):i>4,i\neq 8,9\}\\
\cup\{(i,1):i>1, i \neq 2,3,4,7,8,9\}
\end{multline}
and $\ttcr{\beta'''}=\tm{\beta'''}\cup\tmr{\beta'''}$. The non-zero off-diagonal positions of $\s{10,10,{v'''}^\top}$ belong to this set. Right multiplication by $\s{10,10,{v'''}^\top}$ gives
\begin{multline}
\g{5,u',2}\g{5,u,1}A\s{10,9,v^\top}\s{10,8,{v'}^\top}\s{10,7,{v''}^\top}\s{10,10,{v'''}^\top}\\
=
\begin{pmatrix}
0&0&0&0&0&0&0&0&1&0 \\
0&0&0&0&0&0&0&1&0&0 \\
0&0&0&0&0&0&0&0&0&0 \\
0&0&0&0&0&0&1&0&0&0 \\
0&0&0&0&0&0&0&0&0&1
\end{pmatrix}
\end{multline}
which is the incomplete permutation matrix $\Pi(\alpha''',\beta''')$. Therefore, the canonical form of $A$ is $A=L \Pi(\alpha''', \beta''') R$, with
\begin{equation}
L = \g{5,u,1}\g{5,u',2} 
= 
\begin{pmatrix}
1&0&0&0&0\\
0&1&0&0&0\\
1&1&1&0&0\\
0&0&0&1&0\\
0&0&0&0&1
\end{pmatrix}
\in \L{5}{\tl{\alpha'''}}
\end{equation}
and
\begin{multline}
R=\s{10,10,{v'''}^\top}\s{10,7,{v''}^\top}\s{10,8,{v'}^\top}\s{10,9,v^\top} \\
=
\begin{pmatrix}
1&0&0&0&0&0&0&0&0&0\\
0&1&0&0&0&0&0&0&0&0\\
0&1&1&0&0&0&0&0&0&0\\
0&1&1&1&0&0&0&0&0&0\\
1&1&1&1&1&0&0&0&0&0\\
0&1&1&0&0&1&0&0&0&0\\
1&0&1&0&0&1&1&0&0&0\\
0&1&0&0&1&0&1&1&0&0\\
1&0&0&1&0&0&0&1&1&0\\
0&1&0&1&0&1&0&0&0&1
\end{pmatrix} \\
\in \B{10}{\ttcr{\beta'''}}
\end{multline}

\subsection{Canonical form for the symplectic and Clifford groups}\label{sec:canonical_form_for_symplectic_group}

For an element of $\symplecticgroup{2n}$, it is natural to apply Gaussian elimination with symplectic moves both on the left and on the right, in order to preserve the symplectic products of the rows and of the columns. This algorithm produces a canonical form for $\symplecticgroup{2n}$ in time $O(n^3)$. The tableau version of the algorithm produces a canonical form for $\cliffordgroup{n}$ in time $O(n^3)$. Details follow. 

Since elements of $\symplecticgroup{2n}$ are invertible, they have a pivot in every row. Therefore, the increasing function $\alpha$ that was used previously to describe the row indices of pivots becomes superfluous. Thus, for injective $\beta:[r]\rightarrow[2n]$, let 
\begin{equation}
\Pi(\beta)=\sum_{i=1}^r \e{2n,i}\e{2n,\beta(i)}^{\top}
\end{equation}
and say that $A\in\symplecticgroup{2n}$ is $\beta$ partially reduced if rows $1,\dots, r$ and columns $\beta(1), \dots, \beta(r)$ of $A$ coincide with the corresponding rows and columns of $\Pi(\beta)$. 

As in the case of stabilizer parity check matrices, the symplectic constraint can be used to extract additional information about $\beta$ reduced elements of $\symplecticgroup{2n}$:

\begin{lemma}\label{lemma:PartiallyReducedSymplecticMatrices}
Take $r \leq n$ and let $\beta:[r]\rightarrow[2n]$ be injective. Let $A \in \symplecticgroup{2n}$ be $\beta$ partially reduced, i.e. 
\begin{equation}
\forall i \in [r]: \;\; A \e{2n,\beta(i)} = \e{2n,i} \;\; \wedge \;\; \e{2n,i}^{\top} A=\e{2n,\beta(i)}^{\top}
\end{equation}
Then, $\beta$ is qubit-injective and the following rows and columns of $A$ are also reduced:
\begin{equation}
\forall i \in [r]: \;\; A \e{2n,2n+1-\beta(i)} = \e{2n,2n+1-i} 
\wedge \;\; \e{2n,2n+1-i}^{\top} A = \e{2n,2n+1-\beta(i)}^{\top}
\end{equation}
\end{lemma}
\begin{proof}
$\beta$ is qubit-injective because rows $1, \dots, r$ of $A$ have pairwise symplectic products 0. 

Next, use $A\revdiag{2n}A^{\top}=A^{\top} \revdiag{2n} A=\revdiag{2n}$ and $\e{2n,\beta(i)}^{\top} A^{\top}=\e{2n,i}^{\top}, A^{\top} \e{2n,i} = \e{2n,\beta(i)}$ as follows:
\begin{multline}
A \e{2n,2n+1-\beta(i)} = A\revdiag{2n}\e{2n,\beta(i)} 
=\revdiag{2n}(A^{\top})^{-1} \e{2n,\beta(i)} \\
=\revdiag{2n}\e{2n,i}=\e{2n,2n+1-i} 
\end{multline}
\begin{multline}
\e{2n,2n+1-i}^{\top} A = \e{2n,i}^{\top}\revdiag{2n}A 
=\e{2n,i}^{\top} (A^{\top})^{-1}\revdiag{2n} \\
= \e{2n,\beta(i)}^{\top} \revdiag{2n} = \e{2n,2n+1-\beta(i)}^{\top}
\end{multline}
This completes the proof. 
\end{proof}

This Lemma shows that the pivots in the first $n$ rows uniquely determine all the rest and motivates the following:
\begin{definition}
For $r \leq n$ and qubit-injective $\beta:[r]\rightarrow[2n]$, let
\begin{equation}
\Pi^{sym}(\beta)=\sum_{i=1}^r \left(\e{2n,i}\e{2n,\beta(i)}^{\top} + \e{2n,2n+1-i}\e{2n,2n+1-\beta(i)}^{\top}\right)
\end{equation}
\end{definition}

As before, Gaussian elimination with symplectic moves both on the left and on the right takes a partially reduced matrix and simplifies it further. The typical step of this algorithm is:

\begin{lemma}\label{lemma:FullySymplecticGaussianStep}
Take $r < n$, qubit-injective $\beta:[r]\rightarrow[2n]$, and $\beta$ partially reduced $A \in \symplecticgroup{2n}$. Let $\beta'$ extend $\beta$ to $[r+1]$ so that $(r+1,\beta'(r+1))$ is the position of the next pivot in the left and down order. Let $u,v\in\F^{2n}$ be such that 
\begin{align}
A \e{2n,\beta'(r+1)} &= \e{2n,r+1}+u \\
\e{2n,r+1}^{\top} A & = \e{2n,\beta'(r+1)} + v^{\top}
\end{align}
and let 
\begin{equation}
A'=\s{2n,u,r+1} A \s{2n,\beta'(r+1), v^{\top}}
\end{equation}
Then:
\begin{enumerate}
\item $\s{2n,u,r+1} \in \B{2n}{\orderedpairs{2n}}$. 
\item $\beta'$ is qubit injective. 
\item $\s{2n,\beta'(r+1),v^{\top}} \in \B{2n}{\ttcr{\beta'}}$.
\item $A'\in\symplecticgroup{2n}$ is $\beta'$ partially reduced. 
\item $(\beta',u,v,A')$ can be computed from $(A,\beta)$ in time $O(n^2)$. 
\end{enumerate}
\end{lemma}

\begin{proof}
\emph{Parts 1, 2 and 3:} Since $A$ is $\beta$ partially reduced, Lemma \ref{lemma:PartiallyReducedSymplecticMatrices} implies $\qubit{n}(\beta'(r+1))\notin \{\qubit{n}(\beta(1),\dots,\qubit{n}(\beta(r))\}$, $u$ is supported on a subset of $\{j:r+1<j<2n+1-r\}$, and $v$ is supported on a subset of $\{j:j<\beta'(r+1), \qubit{n}(j) \notin \{\qubit{n}(\beta(1)), \dots \qubit{n}(\beta(r))\}\}$. 

\emph{Part 4:} $A'$ is a product of three elements of $\symplecticgroup{2n}$, so $A'\in\symplecticgroup{2n}$. Next, check the rows: for $i \in [r+1]$:
\begin{equation}
\e{2n,i}^{\top} A' = \e{2n,i}^{\top}\s{2n,u,r+1}A\s{2n,\beta'(r+1),v^{\top}}
=\e{2n,i}^{\top} A \s{2n,\beta'(r+1),v^{\top}} 
\end{equation}
Consider cases:
\begin{enumerate}
\item If $i \leq r $ then
\begin{equation*}
\e{2n,i}^{\top} A \s{2n,\beta'(r+1),v^{\top}} 
= \e{2n,\beta(i)}^{\top}\s{2n,\beta'(r+1),v^{\top}}=\e{2n,\beta(i)}^{\top}
\end{equation*}
\item If $i=r+1 $ then 
\begin{equation*}
\e{2n,i}^{\top} A \s{2n,\beta'(r+1),v^{\top}}  
= (\e{2n,\beta'(r+1)}^{\top} + v^{\top})\s{2n,\beta'(r+1),v^{\top}}=\e{2n,\beta'(r+1)}^{\top}
\end{equation*}
\end{enumerate}
Next, check the columns: for $i \in [r+1]$:
\begin{multline}
A' \e{2n,\beta'(i)} = \s{2n,u,r+1} A \s{2n,\beta'(r+1),v^{\top}} \e{2n,\beta'(i)} 
=\s{2n,u,r+1}A\e{2n,\beta'(i)}\\
=\begin{cases}
\s{2n,u,r+1}\e{2n,i}=\e{2n,i} & \text{if } i \leq r\\
\s{2n,u,r+1}(\e{2n,r+1}+u) = \e{2n,r+1} & \text{if } i=r+1
\end{cases}
\end{multline}

\emph{Part 5:} As before, the sparsity of $\s{2n,u,r+1}$ and $\s{2n,\beta'(r+1),v^{\top}}$ implies that multiplication by them can be performed in time $O(n^2)$. 
\end{proof}

Now, all ingredients are in place to show that Gaussian elimination with symplectic moves on both sides produces a canonical form for $\symplecticgroup{2n}$ in time $O(n^3)$. To state the theorem, it is convenient to introduce the following terminology:
\begin{definition}
Say that $(\beta,L,R)$ is $\symplecticgroup{2n}$ allowed if $\beta:[n]\rightarrow[2n]$ is qubit-injective, $L\in\B{2n}{\orderedpairs{2n}}$ and $R\in\B{2n}{\ttcr{\beta}}$. 
\end{definition}
\begin{theorem}\label{thm:clifford_group_canonical_form}
\begin{enumerate}
\item For every $A \in \symplecticgroup{2n}$, Gaussian elimination with symplectic moves on both sides produces $\symplecticgroup{2n}$ allowed $(\beta,L,R)$ such that 
\begin{equation}
A=L\Pi^{sym}(\beta)R
\end{equation}
Moreover, the algorithm runs in time $O(n^3)$. 
\item If $(\beta,L,R)$ and $(\beta',L',R')$ are $\symplecticgroup{2n}$ allowed and if 
\begin{equation}
L\Pi^{sym}(\beta)R=L'\Pi^{sym}(\beta')R'
\end{equation}
then 
\begin{equation}
(\beta,L,R)=(\beta',L',R')
\end{equation}
\end{enumerate}
\end{theorem}

\begin{proof}
\emph{Part 1} follows from repeated application of Lemma \ref{lemma:FullySymplecticGaussianStep}.

\emph{Part 2:} Lemma \ref{lemma:UniquenessOfIncompletePermutation} implies $\Pi^{sym}(\beta)=\Pi^{sym}(\beta')$, so $\beta=\beta'$. Then, Lemma \ref{lemma:UniquenessOfSymplecticRightFactor}, applied to the first $n$ rows of $(L')^{-1} L \Pi^{sym}(\beta) = \Pi^{sym}(\beta)R'R^{-1}$, gives $R=R'$. $L=L'$ follows. 
\end{proof}

\emph{Example:} Consider the following matrix in $\symplecticgroup{6}$: 
\begin{equation}
A=
\begin{pmatrix}
0&1&1&0&1&0\\
0&0&0&1&1&1\\
0&1&1&0&1&1\\
1&1&0&1&0&0\\
0&0&1&1&1&0\\
1&1&0&1&1&1
\end{pmatrix}
\end{equation}

The first pivot in the left and down order is in the first row and fifth column. Take $\beta(1)=5$. To clear the first row, take 
\begin{equation}
v^\top=\begin{pmatrix} 0&1&1&0&0&0 \end{pmatrix}
\end{equation}
and
\begin{equation}
\s{6,5,v^\top}=
\begin{pmatrix}
1&0&0&0&0&0\\
0&1&0&0&0&0\\
0&0&1&0&0&0\\
0&1&0&1&0&0\\
0&1&1&0&1&0\\
0&0&0&0&0&1
\end{pmatrix}
\end{equation}
At this stage, $\tm{\beta}=\{(5,j):j<5\}$, $\tmr{\beta}=\{(i,2):i>2\}$, $\ttcr{\beta}=\tm{\beta}\cup\tmr{\beta}$. $\s{6,5,v^\top}$ has non-zero off-diagonal entries only in the second column and fifth row. Right multiplication by $\s{6,5,v^\top}$ gives
\begin{equation}
A\s{6,5,v^\top} = 
\begin{pmatrix}
0&0&0&0&1&0\\
0&0&1&1&1&1\\
0&0&0&0&1&1\\
1&0&0&1&0&0\\
0&0&0&1&1&0\\
1&1&1&1&1&1
\end{pmatrix}
\end{equation}
To clear the fifth column, take
\begin{equation}
u=
\begin{pmatrix}0\\1\\1\\0\\1\\1\end{pmatrix}
\; \text{ and } \; 
\s{6,u,1}=
\begin{pmatrix}
1&0&0&0&0&0\\
1&1&0&0&0&0\\
1&0&1&0&0&0\\
0&0&0&1&0&0\\
1&0&0&0&1&0\\
1&1&0&1&1&1
\end{pmatrix}
\end{equation}
Left multiplication by $\s{6,u,1}$ gives
\begin{equation}
\s{6,u,1}A\s{6,5,v^\top} = 
\begin{pmatrix}
0&0&0&0&1&0\\
0&0&1&1&0&1\\
0&0&0&0&0&1\\
1&0&0&1&0&0\\
0&0&0&1&0&0\\
0&1&0&0&0&0
\end{pmatrix}
\end{equation}
At this point, rows 1 and 6 and columns 2 and 5 are reduced. 

The second pivot in the left and down order is in the second row and sixth column. Extend $\beta$ to $\beta'$ with $\beta'(2)=6$. To clear the second row, take 
\begin{equation}
{v'}^\top=\begin{pmatrix} 0&0&1&1&0&0 \end{pmatrix}
\end{equation}
and
\begin{equation}
\s{6,6,{v'}^\top}=
\begin{pmatrix}
1&0&0&0&0&0\\
0&1&0&0&0&0\\
1&0&1&0&0&0\\
1&0&0&1&0&0\\
0&0&0&0&1&0\\
0&0&1&1&0&1
\end{pmatrix}
\end{equation}
At this stage, 
\begin{align}
\tm{\beta'}&=\{(5,j):j<5\}\cup\{(6,j):j<6,j\neq 2,5\}\\
\tmr{\beta'}&=\{(i,2):i>2\}\cup\{(i,1):i>1,i\neq 2,5\}
\end{align}
and $\ttcr{\beta'}=\tm{\beta'}\cup\tmr{\beta'}$. The positions of non-zero off-diagonal entries of $\s{6,6,{v'}^\top}$ belong to this set. Right multiplication by $\s{6,6,{v'}^\top}$ gives
\begin{equation}
\s{6,u,1}A\s{6,5,v^\top}\s{6,6,{v'}^\top} = 
\begin{pmatrix}
0&0&0&0&1&0\\
0&0&0&0&0&1\\
0&0&1&1&0&1\\
0&0&0&1&0&0\\
1&0&0&1&0&0\\
0&1&0&0&0&0
\end{pmatrix}
\end{equation}
To clear the sixth column, take
\begin{equation}
u'=
\begin{pmatrix}0\\0\\1\\0\\0\\0\end{pmatrix}
\; \text{ and } \; 
\s{6,u',2}=
\begin{pmatrix}
1&0&0&0&0&0\\
0&1&0&0&0&0\\
0&1&1&0&0&0\\
0&0&0&1&0&0\\
0&0&0&1&1&0\\
0&0&0&0&0&1
\end{pmatrix}
\end{equation}
Left multiplication by $\s{6,u',2}$ gives
\begin{equation}
\s{6,u'2}\s{6,u,1}A\s{6,5,v^\top}\s{6,6,{v'}^\top} = 
\begin{pmatrix}
0&0&0&0&1&0\\
0&0&0&0&0&1\\
0&0&1&1&0&0\\
0&0&0&1&0&0\\
1&0&0&0&0&0\\
0&1&0&0&0&0
\end{pmatrix}
\end{equation}
At this point, rows 1, 2, 5 and 6 and columns 1, 2, 5 and 6 are reduced. 

The third pivot in the left and down order is in the third row and fourth column. Extend $\beta'$ to $\beta''$ with $\beta''(3)=4$. To clear the third row, take 
\begin{equation}
{v''}^\top= \begin{pmatrix} 0&0&1&0&0&0 \end{pmatrix}
\end{equation}
and
\begin{equation}
\s{6,4,{v''}^\top} = 
\begin{pmatrix}
1&0&0&0&0&0\\
0&1&0&0&0&0\\
0&0&1&0&0&0\\
0&0&1&1&0&0\\
0&0&0&0&1&0\\
0&0&0&0&0&1
\end{pmatrix}
\end{equation}
At this stage, 
\begin{equation}
\tm{\beta''}=\{(5,j):j<5\}\cup\{(6,j):j<6,j\neq 2,5\}
\cup\{(4,3)\}
\end{equation}
\begin{equation}
\tmr{\beta''}=\{(i,2):i>2\}\cup\{(i,1):i>1,i\neq 2,5\}
\cup\{(4,3)\}
\end{equation}
and $\ttcr{\beta''}=\tm{\beta''}\cup\tmr{\beta''}$. The positions of non-zero off-diagonal entries of $\s{6,4,{v''}^\top}$ belong to this set. Right multiplication by $\s{6,4,{v''}^\top}$ gives
\begin{equation}
\s{6,u'2}\s{6,u,1}A\s{6,5,v^\top}\s{6,6,{v'}^\top}\s{6,4,{v''}^\top} 
= 
\begin{pmatrix}
0&0&0&0&1&0\\
0&0&0&0&0&1\\
0&0&0&1&0&0\\
0&0&1&1&0&0\\
1&0&0&0&0&0\\
0&1&0&0&0&0
\end{pmatrix}
\end{equation}
To clear the fourth column, take 
\begin{equation}
u''=
\begin{pmatrix}0\\0\\0\\1\\0\\0\end{pmatrix}
\; \text{ and } \; 
\s{6,u'',3}=
\begin{pmatrix}
1&0&0&0&0&0\\
0&1&0&0&0&0\\
0&0&1&0&0&0\\
0&0&1&1&0&0\\
0&0&0&0&1&0\\
0&0&0&0&0&1
\end{pmatrix}
\end{equation}
Left multiplication by $\s{6,u'',3}$ gives
\begin{equation}
\s{6,u'',3}\s{6,u'2}\s{6,u,1}A\s{6,5,v^\top}\s{6,6,{v'}^\top}\s{6,4,{v''}^\top} 
= 
\begin{pmatrix}
0&0&0&0&1&0\\
0&0&0&0&0&1\\
0&0&0&1&0&0\\
0&0&1&0&0&0\\
1&0&0&0&0&0\\
0&1&0&0&0&0
\end{pmatrix}
\end{equation}
which is $\Pi^{sym}(\beta'')$. Therefore, the canonical form of $A$ is $A=L\Pi^{sym}(\beta'')R$ with 
\begin{equation}
L=\s{6,u,1}\s{6,u',2}\s{6,u'',3}
=
\begin{pmatrix}
1&0&0&0&0&0\\
1&1&0&0&0&0\\
1&1&1&0&0&0\\
0&0&1&1&0&0\\
1&0&1&1&1&0\\
1&1&0&0&1&1
\end{pmatrix}
\in \B{2n}{\orderedpairs{2n}}
\end{equation}
and
\begin{equation}
R=\s{6,4,{v''}^\top}\s{6,6,{v'}^\top}\s{6,5,v^\top}
=
\begin{pmatrix}
1&0&0&0&0&0\\
0&1&0&0&0&0\\
1&0&1&0&0&0\\
0&1&1&1&0&0\\
0&1&1&0&1&0\\
0&1&1&1&0&1
\end{pmatrix}
\in \B{2n}{\ttcr{\beta''}}
\end{equation}

\section{Finite blocklength bounds for stabilizer codes and Pauli noise}\label{sec:finite_blocklength_bounds}

The preparatory subsection \ref{sec:uniform_distribution} establishes a property of the uniform distribution over the symplectic group that is used to derive the achievability bound. Another preparatory subsection \ref{sec:notation_for_sorting} establishes notation necessary to state the bounds. Next, subsection \ref{sec:general_bounds} states and proves achievability and converse bounds on $\errorguesserrorprob{p_{UV},r}$ and $\errorguessrate{p_{UV},\epsilon}$ for arbitrary $p_{UV}$. Subsection \ref{sec:bounds_for_erasure_channel} and subsection \ref{sec:bounds_for_depolarizing_channel} show how the general bounds can be computed efficiently in the cases of the qubit erasure and depolarizing channels respectively. Moreover, it is seen there that the achievability and converse bounds on the rate differ by $O(1/n)$ for $n$ qubits. Subsection \ref{sec:n_independent_identical_channels} shows that the achievability and converse bounds can be computed in polynomial time in the case of $n$ independent identical Pauli channels with side information. 

\subsection{The uniform distribution over symplectic matrices}\label{sec:uniform_distribution}

The following Lemma will be useful later on:

\begin{lemma}\label{lemma:uniform_symplectic_matrix}
Let $C$ be uniformly distributed over $\symplecticgroup{2n}$, and let $u$ be a non-zero vector in $\F^{2n}$. Then, $Cu$ is uniformly distributed over non-zero vectors. 
\end{lemma}

\begin{proof}
Take non-zero vectors $v,v'$. Take symplectic matrix $C'$ such that $C'v=v'$. Then, $\prob{Cu=v}=\prob{C'Cu=v'}=\prob{Cu=v'}$, because $C'C$ is also uniformly distributed over symplectic matrices. 
\end{proof}

The proof of Lemma \ref{lemma:uniform_symplectic_matrix} uses the fact that $\symplecticgroup{2n}$ acts transitively on $\F^{2n}\backslash\{0\}$. Previous expositions of the hashing bound also implicitly establish and use this transitive action; see for example \cite[Section 3.2.3]{bennett1996mixed}. A short proof based on the symplectic Gaussian moves is as follows: 

\begin{lemma}
Let $u,v$ be non-zero vectors in $\F^{2n}$. Then, there exists $C \in \symplecticgroup{2n}$ such that $Cu=v$. 
\end{lemma}

\begin{proof}
If there is $i$ such that $\e{2n,i}^{\top}u=\e{2n,i}^{\top}v=1$, then let $u'=u-\e{2n,i}$, $v'=v-\e{2n,i}$. Then, Lemma \ref{lemma:PropertiesOfSymplecticMoves} implies that $\s{2n,v',i}\s{2n,u',i}u=\s{2n,v',i}\e{2n,i}=v$. 

If there is no $i$ such that $\e{2n,i}^{\top}u=\e{2n,i}^{\top}v=1$, then let $i\neq j$ be such that $u$ has an $i$-th but not a $j$-th component and $v$ has a $j$-th but not an $i$-th component. Then, $\s{2n,j,i}u$ has a $j$-th component, which reduces this case to the previous one. 
\end{proof}

\subsection{Some additional notation}\label{sec:notation_for_sorting}

Take a distribution $p_{UV}$ of a random vector $U$ in $\F^{2n}$ and another discrete random variable $V$. For each value $v$ that $V$ can take, sort the vectors $u \in \F^{2n}$ in decreasing order according to the probability of the event $U=u,V=v$, resolving ties arbitrarily. Summarize this in a funciton $\orderfunction$: for each $v$, $\orderfunction(1,v)$, $\orderfunction(2,v)$, \dots, $\orderfunction(2^{2n},v)$ are the $2^{2n}$ vectors in $\F^{2n}$ sorted so that 
\begin{equation}
p_{UV}(\orderfunction(1,v),v) \geq p_{UV}(\orderfunction(2,v),v) \geq \dots \geq p_{UV}(\orderfunction(2^{2n},v),v)
\end{equation}

Let random variable $J$ taking values in $\left[2^{2n}\right]$ be such that for each $j$ and $v$, the event $J=j,V=v$ is the same as the event $U=\orderfunction(j,v),V=v$.

\subsection{General bounds}\label{sec:general_bounds}

Using the notation in section \ref{sec:notation_for_sorting}, it is possible to formulate the following general bounds: 
 
\begin{theorem}\label{thm:general_bounds}
For every distribution $p_{UV}$ and rate $r=k/n$, let
\begin{align}
\errorguesserrorprobconverse{p_{UV},r}&=\prob{J>2^m} \\
\errorguesserrorprobachievability{p_{UV},r} &=\prob{J>2^m} \nonumber\\
&+ \expect{\indicator{J\leq 2^m} (J-1)2^{-m}}
\end{align}
where $m=n-k$ and $\indicator{\cdot}$ is the indicator of an event. Then,
\begin{equation}
\errorguesserrorprobconverse{p_{UV},r} 
\leq \errorguesserrorprob{p_{UV},r} 
\leq \errorguesserrorprobachievability{p_{UV},r}
\end{equation}

Moreover, let
\begin{equation}
\errorguessrateachievability{p_{UV},\epsilon}
=\max \left\{r:\errorguesserrorprobachievability{p_{UV},r}\leq \epsilon\right\} 
\end{equation}
\begin{equation}
\errorguessrateconverse{p_{UV},\epsilon}
=\min \left\{r: \errorguesserrorprobconverse{p_{UV},r} > \epsilon \right\}
\end{equation}
Then,
\begin{equation}
\errorguessrateachievability{p_{UV},\epsilon}
 \leq \errorguessrate{p_{UV},\epsilon} 
 < \errorguessrateconverse{p_{UV},\epsilon}
\end{equation}
\end{theorem}

\begin{proof}
First, consider the converse bound on the optimal error probability. Without loss of generality, the optimal strategy is deterministic, so there is a fixed symplectic matrix $C$ and the decoder $D$ is a fixed mapping of pairs $(s,v)$ of syndrome and side information to vectors $u\in \F^{2n}$. For every $v$, if $U$ takes a value outside the image of $D$, then an error occurs. In the best case, for every $v$, the $2^m$ possible syndromes are mapped to the $2^m$ most likely vectors $\orderfunction(1,v)$, \dots, $\orderfunction(2^m,v)$. Then, the probability of error is at least
\begin{equation}
\sum_v \sum_{j=2^m+1}^{2^{2n}} p_{UV}(\orderfunction(j,v),v)=\prob{J>2^m}
\end{equation}

Now, consider the achievability bound on the optimal error probability. Let $C$ be uniformly distributed over symplectic matrices. Let the decoder $D$ proceed as follows: try the vectors $\orderfunction(1,v)$, $\orderfunction(2,v)$, $\dots$ in order until one matches the syndrome. Conditional on $J=j,V=v$, the union bound and Lemma \ref{lemma:uniform_symplectic_matrix} imply that the probability that for some $i<j$, $\orderfunction(i,v)$ produces the same syndrome as $\orderfunction(j,v)$ is at most $\min(1,(j-1)2^{-m})$. This proves the achievability bound. 

Finally, bounds on the optimal error probability lead to bounds on the optimal rate.

If $\errorguesserrorprobachievability{p_{UV},r}\leq \epsilon$, then there is a strategy of rate $r$ and error probability at most $\epsilon$, so $\errorguessrate{p_{UV},\epsilon}\geq r$. The best lower bound that can be obtained in this way is $\errorguessrateachievability{p_{UV},\epsilon}$. 

If $\errorguesserrorprobconverse{p_{UV},r} > \epsilon$, then no strategy of rate $r$ has error probability at most $\epsilon$, so $\errorguessrate{p_{UV},\epsilon} < r$. The best upper bound that can be obtained in this way is $\errorguessrateconverse{p_{UV},\epsilon}$. 
\end{proof}

\subsection{The qubit erasure channel}\label{sec:bounds_for_erasure_channel}

As a first example, consider the qubit erasure channel \eqref{eq:erasure_channel}. Let $p_{UV}$ correspond to $n$ independent erasure channels with parameter $\delta$. The side information $V$ is a vector of $n$ independent $Bernoulli(\delta)$ random variables; it specifies the erased qubits. Conditional on $V=v$, $U$ is uniformly distributed over the $2^{2|v|}$ vectors supported on the erased qubits, and $J$ is uniformly distributed on $\left[2^{2|v|}\right]$.

These observations show that for the erasure channel, the general bounds in Theorem \ref{thm:general_bounds} can be expressed in terms of the cumulative distribution function of the binomial distribution. Since the CDF of the binomial can be computed efficiently and only a few calls to this function are needed, this gives an efficient algorithm for computing the bounds. Moreover, the resulting bounds on the rate differ by at most $O(1/n)$ from a simple function of $n$, $\delta$ and $\epsilon$. 

\begin{theorem}\label{thm:erasure_channel_bounds}
Let $\bincdf{n,p,i}$ denote the probability that a $Binomial(n,p)$ random variable is at most $i$:
\begin{equation}\label{eq:binomial_cdf}
\bincdf{n,p,i}=\sum_{i'=0}^i {n \choose i'} p^{i'}(1-p)^{n-i'}
\end{equation}

Let $p_{UV}$ correspond to $n$ independent erasure channels with parameter $\delta$. Then,
\begin{multline}
\errorguesserrorprobconverse{p_{UV},r}=\bincdf{n,1-\delta,n-\floor{\frac{m}{2}}-1}  \\
- 2^m \left(\frac{4-3\delta}{4}\right)^n \bincdf{n,\frac{4-4\delta}{4-3\delta},n-\floor{\frac{m}{2}} -1} 
\end{multline}
and
\begin{multline}
\errorguesserrorprobachievability{p_{UV},r}
=\left(1+\frac{1}{2^{m+1}}\right)\bincdf{n,1-\delta,n-\floor{\frac{m}{2}}-1} - \frac{1}{2^{m+1}} \\
- \frac{2^m+1}{2} \left(\frac{4-3\delta}{4}\right)^n \bincdf{n,\frac{4-4\delta}{4-3\delta},n-\floor{\frac{m}{2}}-1} \\
+\frac{(1+3\delta)^n}{2^{m+1}} \bincdf{n,\frac{4\delta}{1+3\delta},\floor{\frac{m}{2}}} 
\end{multline} 
where $r=k/n$ and $m=n-k$. 

Moreover, for fixed $\delta$, $\epsilon$, and for sufficiently large $n$, 
\begin{equation}
\errorguessrate{p_{UV},\epsilon}=1-2\delta 
+ \frac{2\normalcdfinv{\epsilon}\sqrt{\delta(1-\delta)}}{\sqrt{n}}+O\left(\frac{1}{n}\right)
\end{equation}
where $\normalcdf{x}=1/\sqrt{2\pi} \int_{-\infty}^x e^{-t^2/2}dt$ is the cumulative distribution function of the standard normal distribution. 
\end{theorem}

\begin{proof}
For the converse bound on the error probability,
\begin{equation}
\prob{J>2^m}=\sum_v \prob{V=v}\prob{J>2^m \Big| V=v}
\end{equation}
and $\prob{J>2^m \Big| V=v}$ is zero when $2|v| \leq m$ and is $1-2^{m-2|v|}$ otherwise. This gives
\begin{equation}
\errorguesserrorprobconverse{p_{UV},r}=\expect{\indicator{2|V|>m}(1-2^{m-2|V|})} 
\end{equation}
Next, 
\begin{equation}
\expect{\indicator{2|V|>m}}=\bincdf{n,1-\delta,n-\floor{\frac{m}{2}}-1}
\end{equation}
and
\begin{multline}\label{eq:er_ch_pr_conv_term_2}
-\expect{\indicator{2|V|>m}2^{m-2|V|}} \\
=-2^m \sum_{i=0}^n {n \choose i} \delta^i 4^{-i} (1-\delta)^{n-i} \indicator{2i>m} \\
=-2^m \left(\frac{4-3\delta}{4}\right)^n \sum_{i=0}^n {n \choose i} \left(\frac{\delta}{4-3\delta}\right)^i\left(\frac{4-4\delta}{4-3\delta}\right)^{n-i}\indicator{2i>m} \\
=-2^m \left(\frac{4-3\delta}{4}\right)^n \bincdf{n,\frac{4-4\delta}{4-3\delta},n-\floor{\frac{m}{2}}-1}
\end{multline}

For the achievability bound on the error probability, 
\begin{multline}
\expect{\indicator{J\leq 2^m}(J-1)2^{-m}}\\
=\sum_v \prob{V=v} \expect{\indicator{J\leq 2^m}(J-1)2^{-m}\Big|V=v}
\end{multline}
and
\begin{multline}
\expect{\indicator{J\leq 2^m}(J-1)2^{-m}\Big|V=v}\\
=
\begin{cases}
2^{2|v|-m-1}-2^{-m-1} & \text{if } 2|v| \leq m \\
2^{m-2|v|-1}-2^{-2|v|-1} & \text{if } 2|v| > m
\end{cases}
\end{multline}
This gives
\begin{multline}
\errorguesserrorprobachievability{p_{UV},r} 
=\expect{\indicator{2|V|>m}\left(1-2^{m-2|V|-1}-2^{-2|V|-1}\right)} \\ 
 + \expect{\indicator{2|V|\leq m}\left(2^{2|V|-m-1}-2^{-m-1}\right)} 
\end{multline}
Next,
\begin{multline}
\expect{\indicator{2|V|>m}\left(1-2^{m-2|V|-1}-2^{-2|V|-1}\right)}\\
=\bincdf{n,1-\delta,n-\floor{\frac{m}{2}}-1} \\
- \frac{2^m+1}{2} \left(\frac{4-3\delta}{4}\right)^n \bincdf{n,\frac{4-4\delta}{4-3\delta},n-\floor{\frac{m}{2}}-1}
\end{multline}
and
\begin{multline}\label{eq:er_ch_er_pr_ach}
\expect{\indicator{2|V|\leq m}2^{2|V|-m-1}}
=2^{-m-1} \sum_{i=0}^n {n \choose i} \delta^i 4^i (1-\delta)^{n-i} \indicator{2i\leq m} \\
=2^{-m-1} (1+3\delta)^n \sum_{i=0}^n {n \choose i} \left(\frac{4\delta}{1+3\delta}\right)^i \left(\frac{1-\delta}{1+3\delta}\right)^{n-i} \indicator{2i \leq m} \\
=2^{-m-1}(1+3\delta)^n \bincdf{n,\frac{4\delta}{1+3\delta},\floor{\frac{m}{2}}}
\end{multline}
and
\begin{multline}
-2^{-m-1} \expect{\indicator{2|V|\leq m}} = -2^{-m-1} \bincdf{n,\delta,\floor{\frac{m}{2}}} \\
=-2^{-m-1}\left(1-\bincdf{n,1-\delta,n-\floor{\frac{m}{2}}-1}\right)
\end{multline}

For the converse bound on the rate, return to the expression from equation \eqref{eq:er_ch_pr_conv_term_2}:
\begin{equation}\label{eq:er_ch_pr_conv_sum}
\expect{\indicator{2|V|>m}2^{m-2|V|}} \\
=\sum_{2i>m} {n \choose i} \delta^i  (1-\delta)^{n-i} 2^{m-2i}
\end{equation}
The ratio of the $(i+1)$-st to the $i$-th term is
\begin{equation}
\frac{(n-i)\delta}{4(i+1)(1-\delta)}<\frac{(n-m/2)\delta}{4(m/2+1)(1-\delta)} < 1
\end{equation}
for $m>(2n\delta - 8(1-\delta))/(4-3\delta)$. Therefore, for $m$ in this range the sum in \eqref{eq:er_ch_pr_conv_sum}
is upper bounded by the first binomial probability times a convergent geometric series. Therefore, 
\begin{equation}
\errorguesserrorprobconverse{p_{UV},r} 
= \bincdf{n,1-\delta,n-\floor{\frac{m}{2}}-1} + O\parenth{\frac{1}{\sqrt{n}}}
\end{equation}
Next, the Berry-Esseen theorem \cite{berry1941accuracy,esseen1942liapounoff} gives
\begin{equation}
\errorguesserrorprobconverse{p_{UV},r} = \normalcdf{\frac{n\delta-\floor{m/2}-1}{\sqrt{n\delta(1-\delta)}}} + O\parenth{\frac{1}{\sqrt{n}}}
\end{equation}
Finally, Taylor expansion of $\normalcdfinv{\epsilon+O\parenth{1/\sqrt{n}}}$ shows that choosing 
\begin{equation}
m=2n\delta - 2\sqrt{n\delta(1-\delta)}\normalcdfinv{\epsilon}+O(1)
\end{equation}
suffices to ensure $\errorguesserrorprobconverse{p_{UV},r}>\epsilon$. 

For the achievability bound on the rate, 
\begin{multline}
\left(1+\frac{1}{2^{m+1}}\right)\bincdf{n,1-\delta,n-\floor{\frac{m}{2}}-1} - \frac{1}{2^{m+1}} \\
- \frac{2^m+1}{2} \left(\frac{4-3\delta}{4}\right)^n \bincdf{n,\frac{4-4\delta}{4-3\delta},n-\floor{\frac{m}{2}}-1}  \\ \leq \normalcdf{\frac{n\delta-\floor{m/2}-1}{\sqrt{n\delta(1-\delta)}}} + O\parenth{\frac{1}{\sqrt{n}}}
\end{multline}
similarly to the argument for the converse bound. Moreover, the remaining term from equation \eqref{eq:er_ch_er_pr_ach} is
\begin{equation}\label{eq:er_ch_er_pr_ach_remaining_term}
\expect{\indicator{2|V|\leq m}\left(2^{2|V|-m-1}\right)} 
=\sum_{2i\leq m} {n \choose i} \delta^i (1-\delta)^{n-i}2^{2i-m-1}
\end{equation}
The ratio of the $(i-1)$-st to the $i$-th term is 
\begin{equation}
\frac{i(1-\delta)}{4(n-i-1)\delta}\leq \frac{(m/2)(1-\delta)}{4\delta(n-m/2+1)} < 1
\end{equation}
when $m < 8\delta(n+1)/(1+3\delta)$. Therefore, for $m$ in this range, the sum in \eqref{eq:er_ch_er_pr_ach_remaining_term} is upper bounded by $O(1/\sqrt{n})$, again by the binomial probability times convergent geometric series argument. Thus, for $m < 8\delta(n+1)/(1+3\delta)$, 
\begin{equation}
\errorguesserrorprobachievability{p_{UV},r}
\leq \normalcdf{\frac{n\delta-\floor{m/2}-1}{\sqrt{n\delta(1-\delta)}}} + O\parenth{\frac{1}{\sqrt{n}}}
\end{equation}
Therefore, choosing 
\begin{equation}
m=2n\delta - 2\sqrt{n\delta(1-\delta)}\normalcdfinv{\epsilon}+O(1)
\end{equation}
suffices to ensure $\errorguesserrorprobachievability{p_{UV},r}\leq \epsilon$. This completes the proof. 
\end{proof}

\subsection{The qubit depolarizing channel}\label{sec:bounds_for_depolarizing_channel}

As a second example, consider the qubit depolarizing channel \eqref{eq:depolarizing_channel}. Consider $n$ independent qubit depolarizing channels with parameter $\delta$. In this case, there is no side information. The distribution of $U$ is $\prob{U=u}=(\delta/3)^{|u|} (1-\delta)^{n-|u|}$, where $|u|$ is the number of qubits for which the corresponding two entries of $u$ are not $00$. Conditional on $|U|=i$, $J$ is a uniformly distributed integer between 
\begin{equation}
\left|\left\{u:|u|\leq i-1\right\}\right|+1 =4^n \bincdf{n,\frac{3}{4},i-1} +1
\end{equation}
and
\begin{equation}
\left|\left\{u:|u|\leq i\right\}\right| =4^n \bincdf{n,\frac{3}{4},i}
\end{equation}
where the notation for the binomial CDF from equation \eqref{eq:binomial_cdf} is used. 

These observations show that the cumulative distribution function of $J$ can be computed in terms of the CDF of the binomial.

\begin{lemma}
For every $n,p$, extend the binomial CDF $\bincdf{n,p,\cdot}$ from equation \eqref{eq:binomial_cdf} to a piecewise linear function whose graph connects the points $(-1,0)$, $(0,\bincdf{n,p,0})$, \dots, $(n,\bincdf{n,p,n})$. This extension is an increasing bijection that maps $[-1,n]$ to $[0,1]$; let $\bincdfinv{n,p,\cdot}$ be the inverse function. 

With this notation, the CDF of $J$ for $n$ independent depolarizing channels with parameter $\delta$ is 
\begin{equation}
\prob{J\leq j}=\bincdf{n,\delta,\bincdfinv{n,\frac{3}{4},\frac{j}{4^n}}}
\end{equation}
\end{lemma}

\begin{proof}
For each $i=0$, \dots, $n$, the events $J/4^n \leq \bincdf{n,3/4,i}$ and $|U|\leq i$ coincide. Therefore, the CDF of $J/4^n$ can be extended to a piecewise linear function $[0,1] \rightarrow [0,1]$ whose graph connects the points $(0,0)$, $(\bincdf{n,3/4,0},\bincdf{n,\delta,0})$, \dots, $(\bincdf{n,3/4,n},\bincdf{n,\delta,n})$. This can equivalently be expressed as the composition of two piecewise linear functions: $\bincdfinv{n,3/4,\cdot}$, whose graph connects the points $(0,-1)$, $(\bincdf{n,3/4,0},0)$, \dots, $(\bincdf{n,3/4,n},n)$, and $\bincdf{n,\delta,\cdot}$, whose graph connects the points $(-1,0)$, $(0,\bincdf{n,\delta,0})$, \dots, $(n,\bincdf{n,\delta,n})$.
\end{proof}

Knowing the CDF of $J$ in terms of the CDF of the binomial gives an efficient algorithm for computing the general bounds of Theorem \ref{thm:general_bounds} in the case of $n$ independent depolarizing channels. Moreover, the resulting bounds on the rate differ by at most $O(1/n)$ from a simple function of $n$, $\delta$ and $\epsilon$. 

\begin{theorem}\label{thm:depolarizing_channel_bounds}
Let $p_U$ correspond to $n$ indepedent depolarizing channels with parameter $\delta$. Then,
\begin{equation}
\errorguesserrorprobconverse{p_U,r} = \bincdf{n,1-\delta,n-1-\ell}
\end{equation}
and
\begin{multline}
\errorguesserrorprobachievability{p_U,r}
 = \left(1+\frac{1}{2^{m+1}}\right)\bincdf{n,1-\delta,n-1-\ell} \\
 - \frac{1}{2^{m+1}}
+2^{m-1}(1-\delta)^n\fracexp{\delta}{3-3\delta}{\floor{\ell}+1}\\
+\frac{(16-16\delta)^n}{2^{m+1}}\parenth{\frac{3-4\delta}{3-3\delta}}\sum_{i=0}^{\floor{\ell}}\fracexp{\delta}{3-3\delta}{i}\bincdf{n,\frac{3}{4},i}^2
\end{multline}
where $r=k/n$, $m=n-k$ and $\ell=\bincdfinv{n,\frac{3}{4},\frac{2^m}{4^n}}$. 

Moreover, for fixed $\delta$, $\epsilon$ and for sufficiently large $n$, 
\begin{multline}
\errorguessrate{p_U,\epsilon}=1-h(\delta)-\delta\log_2(3) \\
-\sqrt{\frac{\delta(1-\delta)}{n}} \normalcdfinv{\epsilon} \logpar{\frac{\delta}{3(1-\delta)}} \\
+ \frac{\log_2(n)}{2n}+O\parenth{\frac{1}{n}}
\end{multline}
where $h(\cdot)$ is the binary entropy. 
\end{theorem}

\begin{proof}
For the converse bound on the error probability, 
\begin{multline}
\prob{J>2^m}=1-\prob{J\leq 2^m}
=1-\bincdf{n,\delta,\bincdfinv{n,\frac{3}{4},\frac{2^m}{4^n}}}\\=1-\bincdf{n,\delta,\ell}=\bincdf{n,1-\delta,n-1-\ell}
\end{multline}

For the achievability bound on the error probability, 
\begin{multline}
\prob{J>2^m}+\expect{\indicator{J\leq 2^m} (J-1)2^{-m}}\\
=\prob{J>2^m} - \frac{1}{2^{m+1}}\prob{J\leq 2^m} 
+ \frac{1}{2^m} \expect{\indicator{J\leq 2^m}\left(J-\frac{1}{2}\right)}\\
=\left(1+\frac{1}{2^{m+1}}\right)\bincdf{n,1-\delta,n-1-\ell}-\frac{1}{2^{m+1}} \\
+\frac{1}{2^m}\expect{\indicator{J\leq 2^m}\left(J-\frac{1}{2}\right)}
\end{multline}
Next, partition the event $J\leq 2^m$ into the disjoint events
\begin{equation}
4^n \bincdf{n,\frac{3}{4},i-1} < J  \leq 4^n \bincdf{n,\frac{3}{4},i}, \;\; i=0,\dots,\floor{\ell}
\end{equation}
and 
\begin{equation}
4^n\bincdf{n,\frac{3}{4},\floor{\ell}} < J \leq 4^n\bincdf{n,\frac{3}{4},\ell}=2^m
\end{equation}
and note that for $i=0$, \dots, $\floor{\ell}$, 
\begin{multline}
\expect{\indicator{\bincdf{n,\frac{3}{4},i-1}<\frac{J}{4^n}\leq  \bincdf{n,\frac{3}{4},i}}\left(J-\frac{1}{2}\right)}\\
=\prob{4^n\bincdf{n,\frac{3}{4},i-1}<J\leq 4^n\bincdf{n,\frac{3}{4},i}} \\
*\expect{J-\frac{1}{2}\Bigg|4^n\bincdf{n,\frac{3}{4},i-1}<J\leq 4^n\bincdf{n,\frac{3}{4},i}}\\
= \parenth{4^n\bincdf{n,\frac{3}{4},i}-4^n\bincdf{n,\frac{3}{4},i-1}} \fracexp{\delta}{3}{i}(1-\delta)^{n-i} \\
*\frac{4^n\bincdf{n,3/4,i-1}+4^n\bincdf{n,3/4,i}}{2} \\
=\frac{(16-16\delta)^n}{2} \fracexp{\delta}{3-3\delta}{i} 
\parenth{\bincdf{n,\frac{3}{4},i}^2-\bincdf{n,\frac{3}{4},i-1}^2}
\end{multline}
and, similarly, 
\begin{multline}
\expect{\indicator{\bincdf{n,\frac{3}{4},\floor{\ell}}<\frac{J}{4^n}\leq \bincdf{n,\frac{3}{4},\ell}}\parenth{J-\frac{1}{2}}}\\
=\frac{(16-16\delta)^n}{2} \fracexp{\delta}{3-3\delta}{\floor{\ell}+1} \parenth{\bincdf{n,\frac{3}{4},\ell}^2-\bincdf{n,\frac{3}{4},\floor{\ell}}^2}
\end{multline}
Rearrange slightly to conclude that 
\begin{multline}\label{eq:dep_ch_ach_correction}
\frac{1}{2^m}\expect{\indicator{J\leq 2^m}\parenth{J-\frac{1}{2}}} \\= \frac{(16-16\delta)^n}{2^{m+1}} \fracexp{\delta}{3-3\delta}{\floor{\ell}+1}\bincdf{n,\frac{3}{4},\ell}^2 \\
+ \frac{(16-16\delta)^n}{2^{m+1}} \parenth{\frac{3-4\delta}{3-3\delta}} \sum_{i=0}^{\floor{\ell}}  \fracexp{\delta}{3-3\delta}{i}\bincdf{n,\frac{3}{4},i}^2
\end{multline}
and use $\bincdf{n,3/4,\ell}=2^m/4^n$ in the first term.

For the converse bound on the rate, note that the Berry-Esseen theorem implies
\begin{equation}
\errorguesserrorprobconverse{p_U,r} = \normalcdf{\frac{n\delta-1-\ell}{\sqrt{n\delta(1-\delta)}}} + O\parenth{\frac{1}{\sqrt{n}}}
\end{equation}
Next, Taylor expansion of $\normalcdfinv{\epsilon+O\parenth{1/\sqrt{n}}}$ shows that choosing 
\begin{equation}
\ell = n\delta - \sqrt{n\delta(1-\delta)}\normalcdfinv{\epsilon}+O(1)
\end{equation}
suffices to ensure $\errorguesserrorprobconverse{p_{U},r}>\epsilon$. According to Lemma \ref{lemma:estimates_for_dep_ch} below, this choice of $\ell$ corresponds to 
\begin{multline}
m=n(h(\delta)+\delta\log_2(3))\\
+\sqrt{n\delta(1-\delta)}\normalcdfinv{\epsilon}\logpar{\frac{\delta}{3(1-\delta)}}\\
-\frac{1}{2}\log_2(n) + O(1)
\end{multline}

For the achievability bound on the rate, \eqref{eq:upper_bound_on_dep_ch_ach_correction} in Lemma \ref{lemma:estimates_for_dep_ch} combined with the Berry-Esseen theorem and Taylor expansion of $\normalcdfinv{\epsilon+O(1/\sqrt{n})}$ imply that choosing 
\begin{equation}
\ell = n\delta - \sqrt{n\delta(1-\delta)}\normalcdfinv{\epsilon}+O(1)
\end{equation}
suffices to ensure that 
\begin{multline}
\errorguesserrorprobachievability{p_U,r} \\ 
=\prob{J>2^m} - \frac{1}{2^{m+1}}\prob{J\leq 2^m}
+ \frac{1}{2^m} \expect{\indicator{J\leq 2^m}\left(J-\frac{1}{2}\right)}\\
\leq \prob{J>2^m}  + \frac{1}{2^m} \expect{\indicator{J\leq 2^m}\left(J-\frac{1}{2}\right)}\\
\leq \bincdf{n,1-\delta,n-1-\ell} +O\parenth{\frac{1}{\sqrt{n}}} \leq \epsilon
\end{multline}
As before, Lemma \ref{lemma:estimates_for_dep_ch} implies that this choice of $\ell$ corresponds to 
\begin{multline}
m=n(h(\delta)+\delta\log_2(3))\\
+\sqrt{n\delta(1-\delta)}\normalcdfinv{\epsilon}\logpar{\frac{\delta}{3(1-\delta)}}\\
-\frac{1}{2}\log_2(n) + O(1)
\end{multline}
This completes the proof. 
\end{proof}

The next Lemma establishes estimates that are used in the proof of Theorem \ref{thm:depolarizing_channel_bounds}.

\begin{lemma}\label{lemma:estimates_for_dep_ch}
For fixed $\delta$, $\epsilon$ and for sufficiently large $n$, if
\begin{equation}
\bincdf{n,\frac{3}{4},\ell}=\frac{2^m}{4^n}
\end{equation}
and
\begin{equation}
\ell = n\delta - \sqrt{n\delta(1-\delta)}\normalcdfinv{\epsilon}+O(1)
\end{equation}
then
\begin{multline}\label{eq:estimate_for_m}
m=n(h(\delta)+\delta\log_2(3))\\
+\sqrt{n\delta(1-\delta)}\normalcdfinv{\epsilon}\logpar{\frac{\delta}{3(1-\delta)}}\\
-\frac{1}{2}\log_2(n) + O(1)
\end{multline}

Moreover, 
\begin{equation}\label{eq:upper_bound_on_dep_ch_ach_correction}
\frac{1}{2^m}\expect{\indicator{J\leq 2^m}\parenth{J-\frac{1}{2}}} \leq O\parenth{\frac{1}{\sqrt{n}}}
\end{equation}
where the implied constant does not depend on $\ell$.
\end{lemma}

\begin{proof}
\textbf{Claim 1:}
\begin{equation}\label{eq:estimates_for_dep_ch_claim1}
{n \choose \floor{\ell}}\fracexp{3}{4}{\floor{\ell}}\fracexp{1}{4}{n-\floor{\ell}} 
\leq \bincdf{n,\frac{3}{4},\ell} 
 \leq \frac{3}{2}{n \choose \ceiling{\ell}}\fracexp{3}{4}{\ceiling{\ell}}\fracexp{1}{4}{n-\ceiling{\ell}}
\end{equation}
\textbf{Proof of Claim 1:} The lower bound follows immediately. For the upper bound, note that
\begin{equation}
\bincdf{n,\frac{3}{4},\ell}\leq \sum_{i=0}^{\ceiling{\ell}} {n \choose i}\fracexp{3}{4}{i}\fracexp{1}{4}{n-i}
\end{equation}
and the ratio of the $(i-1)$-st to the $i$-th term in the sum is 
\begin{equation}\label{eq:ratio_of_bin_prob}
\frac{{n \choose i-1}\fracexp{3}{4}{i-1}\fracexp{1}{4}{n-i+1}}{{n \choose i}\fracexp{3}{4}{i}\fracexp{1}{4}{n-i}}=\frac{i}{3(n-i+1)}\leq\frac{2\delta}{3(1-2\delta)} \leq \frac{1}{3}
\end{equation}
for the given range of $\ell$ and for $\delta \leq 1/4$ (which can be assumed without loss of generality because the hashing bound is negative already at $\delta=1/5$). Therefore, the sum can be upper bounded by the last term times a convergent geometric series. 

\bigskip
\textbf{Claim 2:}
\begin{equation}\label{eq:estimates_for_dep_ch_claim2}
\log_2 \bincdf{n,\frac{3}{4},\ell} = -n\relent{\frac{\ell}{n}}{\frac{3}{4}}-\frac{1}{2}\log_2(n) + O(1)
\end{equation}
where $\relent{\cdot}{\cdot}$ is the relative entropy. 

\textbf{Proof of Claim 2:} Follows from \eqref{eq:estimates_for_dep_ch_claim1} and the estimate \cite[Chapter 10, Lemma 7]{macwilliams1977theory}:
\begin{equation}
\log_2 {n \choose \lambda n} = nh(\lambda) -\frac{1}{2} \log_2(n) - \frac{1}{2}\log_2(\lambda(1-\lambda)) + O(1)
\end{equation}
Note also that for the given range of $\ell$, the $\log_2(\lambda(1-\lambda))$ can be absorbed in the $O(1)$. Similarly, the differences between $\ell$, $\floor{\ell}$ and $\ceiling{\ell}$ amount to no more than $O(1)$. 

\bigskip
\textbf{Claim 3:}
\begin{multline}\label{eq:estimates_for_dep_ch_claim3}
\relent{\frac{\ell}{n}}{\frac{3}{4}} = \relent{\delta}{\frac{3}{4}} \\
-\sqrt{\frac{\delta(1-\delta)}{n}}\normalcdfinv{\epsilon}\logpar{\frac{\delta}{3(1-\delta)}} + O\parenth{\frac{1}{n}}
\end{multline}

\textbf{Proof of Claim 3:} Follows from the Taylor expansion
\begin{equation}
\relent{x}{y}
=\relent{x_0}{y}+(x-x_0)\logpar{\frac{x_0(1-y)}{(1-x_0)y}} + O\parenth{(x-x_0)^2}
\end{equation}
and the given estimate for $\ell$. 

\bigskip
\textbf{Claim 4:} \eqref{eq:estimates_for_dep_ch_claim2} and \eqref{eq:estimates_for_dep_ch_claim3} and the identity $2-\relent{\delta}{3/4}=h(\delta)+\delta\log_2(3)$ prove \eqref{eq:estimate_for_m}. 

\textbf{Proof of Claim 4:} follows by basic rearrangement. 

\bigskip
\textbf{Claim 5:}
\begin{equation}\label{eq:estimates_for_dep_ch_claim5}
\sum_{i=0}^{\floor{\ell}}  \fracexp{\delta}{3-3\delta}{i}\bincdf{n,\frac{3}{4},i}^2 
\leq \frac{27}{11}\fracexp{\delta}{3-3\delta}{\floor{\ell}}\bincdf{n,\frac{3}{4},\floor{\ell}}^2
\end{equation}

\textbf{Proof of Claim 5:} From \eqref{eq:ratio_of_bin_prob} deduce that
\begin{equation}
\frac{\bincdf{n,3/4,i-1}}{\bincdf{n,3/4,i}} \leq \frac{2\delta}{3(1-2\delta)}
\end{equation}
then bound the sum by the last term times a convergent geometric series with ratio
\begin{equation}
\frac{3(1-\delta)}{\delta}\frac{4\delta^2}{9(1-2\delta)^2}\leq \frac{16}{27}
\end{equation}
for $\delta \leq 1/5$ (which can be assumed without loss of generality because the hashing bound is negative at $\delta=1/5$). 

\bigskip
\textbf{Claim 6:}
\begin{multline}
\frac{1}{2^m}\expect{\indicator{J\leq 2^m}\parenth{J-\frac{1}{2}}}\\ 
\leq \frac{(16-16\delta)^n}{2^{m+1}}\fracexp{\delta}{3-3\delta}{\floor{\ell}+1}\\
*\bincdf{n,\frac{3}{4},\ell}^2\parenth{1+\frac{27}{11}\cdot\frac{3-4\delta}{\delta}} \\
\leq\parenth{\frac{1}{2}+\frac{27(3-4\delta)}{22\delta}}2^m (1-\delta)^{n-\floor{\ell}-1}\fracexp{\delta}{3}{\floor{\ell}+1}\\
\leq\parenth{\frac{1}{2}+\frac{27(3-4\delta)}{22\delta}} \frac{2^m}{4^n}\frac{{n\choose \floor{\ell}+1}\delta^{\floor{\ell}+1}(1-\delta)^{n-\floor{\ell}-1}}{{n\choose \floor{\ell}+1}\fracexp{3}{4}{\floor{\ell}+1}\fracexp{1}{4}{n-\floor{\ell}-1}}\\
\leq \parenth{\frac{1}{2}+\frac{27(3-4\delta)}{22\delta}} \frac{3}{2} {n\choose \floor{\ell}+1}\delta^{\floor{\ell}+1}(1-\delta)^{n-\floor{\ell}-1}\\
\leq O\parenth{\frac{1}{\sqrt{n}}}
\end{multline}
where the implied constant in the last line does not depend on $\ell$. This proves \eqref{eq:upper_bound_on_dep_ch_ach_correction}. 

\textbf{Proof of Claim 6:} The first step follows from \eqref{eq:estimates_for_dep_ch_claim5} and \eqref{eq:dep_ch_ach_correction}. The second and third step follow by rearrangement. The fourth step follows from \eqref{eq:estimates_for_dep_ch_claim1}. The last step follows from the upper bound $O(1/\sqrt{n})$ on individual binomial probabilities. 
\end{proof}

\subsection{Finite blocklength bounds for $n$ independent identical Pauli channels}\label{sec:n_independent_identical_channels}

This subsection shows that the achievability and converse bounds of Theorem \ref{thm:general_bounds} can be computed in polynomial time in the case of $n$ independent identical channels. After some preliminaries, the result is stated in Theorem \ref{thm:computing_the_bounds} and proved in a sequence of Lemmas. 

\subsubsection{Preliminaries}\label{subsec:background_for_n_independent_identical_channels}

\paragraph{A Lemma about orbits}

\begin{lemma}\label{lemma:stabilizer_subgroup_over_an_orbit}
Let $G$ be a finite group that acts on a finite set $T$. For $t \in T$, let $G_t=\{g \in G : gt=t\}$ be the subgroup that stabilizes $t$. If $t'=gt$ for some $g \in G$, then $G_{t'} = g G_t g^{-1}$
\end{lemma}

\begin{proof}
$h \in G_{t}$ is equivalent to $ht=t$, which is equivalent to $ghg^{-1}gt=gt$, which is equivalent to $ghg^{-1}t'=t'$, which is equivalent to $ghg^{-1} \in G_{t'}$. 
\end{proof}

\begin{lemma}\label{lemma:orbits_of_the_diagonal_action}
Let $G$ be a finite group that acts on finite sets $S,T$. Consider the action of $G$ on $S \times T$ given by $g(s,t) = (gs,gt)$. Take any $(s,t) \in S \times T$. Partition the orbit $G(s,t)$ of $(s,t)$ according to the second element:
\begin{equation}
G(s,t) = \cup_{t' \in Gt} \parenth{\parenth{S \times \{t'\}} \cap G(s,t)}
\end{equation}
Then, the size of all sets in the partition is the size of the orbit $G_t s$ of $s$ under the subgroup that stabilizes $t$. 
\end{lemma}

\begin{proof}
The elements of the orbit $G(s,t)$ that have second element $t$ have the form $(gs,t)$ with $g \in G_t$. Therefore, $|(S\times\{t\})\cap G(s,t)| = |G_ts|$. 

Now, take any other $(s',t')=g(s,t) \in G(s,t)$. Similarly to before, deduce $|(S\times\{t'\}) \cap G(s',t')| = |G_{t'} s'|$. 

Next, $G(s',t')=G(s,t)$. Moreover, Lemma \ref{lemma:stabilizer_subgroup_over_an_orbit} gives $|G_{t'} s'| = |g G_t g^{-1} gs| = |g G_t s| = |G_t s|$. This completes the proof. 
\end{proof}

\paragraph{Types and type classes}

For a finite set $S$, let $\simplex{S}$ be the set of probability mass functions on $S$. For $n \in \mathbb{N}$, let $\discretesimplex{S}{n}$ be the discrete subset of $\simplex{S}$ where each element has probability of the form $k/n, k \in \{0,1,\dots,n\}$. The number of elements of $\discretesimplex{S}{n}$ is ${n+|S|-1 \choose |S|-1}$.

The type of an $n$-tuple $s=(s_1,\dots, s_n) \in S^n$ is the probability mass function $\type{s} \in \discretesimplex{S}{n}$ that assigns to each element $s' \in S$ the probability $\type{s}(s')=\sum_{i=1}^n \indicator{s_i=s'} /n$. The type class of a probability mass function $q \in \discretesimplex{S}{n}$ is the subset $\typeclass{q} \subset S^n$ consisting of all $n$ tuples whose type is $q$. The number of elements of the type class $\typeclass{q}$ is $n!/\left(\prod_{s' \in S} (nq(s'))!\right)$. The collection of all type classes $\{\typeclass{q}:q \in \discretesimplex{S}{n}\}$ is a partition of $S^n$. 

\subsubsection{Computing the bounds}\label{subsec:computing_the_bounds}

\begin{theorem}\label{thm:computing_the_bounds}
Take a subset $\supportu$ of $\F^{2a}$. Take a finite set $\supportv$. Take $p \in \simplex{\supportu\times\supportv}$; $p$ specifies a Pauli channel with side information $\paulich{p}$ acting on $a$ qubits (as in equation \eqref{eq:generic_pauli_channel}). 

Let $\tilde{p}_{UV}(u,v) = \prod_{i=1}^n p(u_i,v_i)$; $\tilde{p}$ specifies the Pauli channel with side information $\paulich{\tilde{p}}=\paulich{p}^{\otimes n}$ acting on $N=na$ qubits. Then:
\begin{enumerate}
\item $\errorguesserrorprobconverse{\tilde{p}_{UV},r}$ and $\errorguesserrorprobachievability{\tilde{p}_{UV},r}$ for all $r=k/N$, $k=1,\dots,N$ can be computed in time $O\parenth{n^{|\supportu||\supportv|}}$. 
\item The marginal cost of computing $\errorguessrateconverse{\tilde{p}_{UV},\epsilon}$ and $\errorguessrateachievability{\tilde{p}_{UV},\epsilon}$ for a given $\epsilon >0$ is $O\parenth{\log(n)}$.
\end{enumerate}
\end{theorem}

The proof of Theorem \ref{thm:computing_the_bounds} has a number of steps. 

First, given the points where the CDF of $J$ changes slope, the bounds of Theorem \ref{thm:general_bounds} can be computed as follows: 

\begin{lemma}\label{lemma:computing_the_bounds_from_points}
Let the CDF of $J$ be a piecewise linear function whose graph connects the points $(\alpha_0,\beta_0)=(0,0)$, $(\alpha_i,\beta_i)$ for $i=1, \dots, L$, and $(\alpha_{L+1},\beta_{L+1})=(4^N,1)$, and suppose the points $(\alpha_i,\beta_i), i=1,\dots,L$ have already been computed. Then:
\begin{enumerate}
\item $\errorguesserrorprobconverse{\tilde{p}_{UV},r}$ for all $r=k/N$, $k=1,\dots,N$ can be computed in time $O\parenth{N\log L}$. 
\item $\errorguesserrorprobachievability{\tilde{p}_{UV},r}$ for all $r=k/N$, $k=1,\dots,N$ can be computed in time $O\parenth{N L}$. 
\item The marginal cost of computing $\errorguessrateconverse{\tilde{p}_{UV},\epsilon}$ for a given $\epsilon >0$ is $O\parenth{\log(N)}$.
\item The marginal cost of computing $\errorguessrateachievability{\tilde{p}_{UV},\epsilon}$ for a given $\epsilon >0$ is $O\parenth{\log(N)}$.
\end{enumerate}
\end{lemma}

\begin{proof}
\emph{Part 1:} For each $k$, let $m=N-k$. First, find $i$ such that $\alpha_i \leq 2^m < \alpha_{i+1}$; this takes time $O\parenth{\log L}$ using binary search. Then, compute
\begin{equation}
\prob{J > 2^m} = 1- \prob{J \leq 2^m} 
= 1 - \parenth{\beta_i + (2^m - \alpha_i) \frac{\beta_{i+1}-\beta_i}{\alpha_{i+1}-\alpha_i}}
\end{equation}
This takes time $O(1)$. 

\emph{Part 2:} Continuing from Part 2, note that
\begin{multline}
\expect{\indicator{J\leq 2^m} (J-1) 2^{-m}} \\
= \expect{\indicator{\alpha_i < J \leq 2^m}(J-1)2^{-m}}\\
+ \sum_{i'=0}^{i-1} \expect{\indicator{\alpha_{i'} < J \leq \alpha_{i'+1}} (J-1) 2^{-m}} \\
=\prob{\alpha_i < J \leq 2^m} \expect{(J-1)2^{-m} | \alpha_i < J \leq 2^m} \\
+ \sum_{i'=0}^{i-1} \prob{\alpha_{i'} < J \leq \alpha_{i'+1}} \expect{(J-1)2^{-m}|\alpha_{i'} < J \leq \alpha_{i'+1}} \\
= (2^m - \alpha_i)\frac{\beta_{i+1}-\beta_i}{\alpha_{i+1}-\alpha_i} \frac{\alpha_i + 2^m-1}{2^{m+1}} \\
+ \sum_{i'=0}^{i-1} (\beta_{i'+1}-\beta_{i'}) \frac{\alpha_{i'} + \alpha_{i'+1} -1}{2^{m+1}}
\end{multline}
and this can be computed in time $O\parenth{L}$. 

\emph{Part 3:} Given $\errorguesserrorprobconverse{p_U,r}$ for all $r$, the smallest rate such that $\errorguesserrorprobconverse{p_U,r}>\epsilon$ can be found in time $O\parenth{\log N}$ by binary search. 

\emph{Part 4:} Given $\errorguesserrorprobachievability{p_U,r}$ for all $r$, the largest rate such that $\errorguesserrorprobconverse{p_U,r} \leq \epsilon$ can be found in time $O\parenth{\log N}$ by binary search. 
\end{proof}

Therefore, it suffices to compute the points where the CDF of $J$ changes slope. Equivalently, it suffices to compute the $x$ coordinates and the slope of each line segment: 

\begin{lemma}\label{lemma:points_from_interval_slope}
Let $0=\alpha_0 < \alpha_1 < \dots < \alpha_L < \alpha_{L+1}=4^N$ and $\gamma_1, \gamma_2, \dots, \gamma_{L+1}$ be such that the CDF of $J$ has slope $\gamma_i$ on the interval $(\alpha_{i-1}, \alpha_i)$. Then, $\beta_0, \dots, \beta_{L+1}$ as in Lemma \ref{lemma:computing_the_bounds_from_points} can be computed in time $O(L)$. 
\end{lemma}

\begin{proof}
Let $\beta_0 = 0$ and for $i=1,\dots L+1$, let $\beta_i = \beta_{i-1} + \gamma_i (\alpha_i - \alpha_{i-1})$. 
\end{proof}

The representation of the CDF of $J$ via the pairs $\alpha_i,\gamma_i$ as in Lemma \ref{lemma:points_from_interval_slope} may be called an interval-slope representation of a piecewise linear function. It is convenient when computing a sum of such functions: 

\begin{lemma}\label{lemma:sum_of_piecewise_linear_functions}
Let $f,f'$ be piecewise linear functions on $[0,4^N]$ such that $f(0)=f'(0)=0$ and whose interval-slope representations are $\alpha_i,\gamma_i, i=1,\dots L+1$ and $\alpha'_{i'},\gamma'_{i'}, i=1,\dots L'+1$. Let $f''= f + f'$. Then, the interval-slope representation of $f''$ contains at most $L+L'+1$ pairs $\alpha''_{i''},\gamma''_{i''}$ and can be computed in time $O\parenth{L+L'}$. 

More generally, for all $i$ in some finite range, let $f_i$ be a piecewise linear function on $[0,4^N]$ with $f_i(0)=0$ and with $L_i+1$ pairs in its interval-slope representation. Then, the interval-slope representation of the sum $\sum_i f_i$ has at most $1+ \sum_i L_i$ pairs and can be computed in time $O\parenth{\sum_i L_i}$. 
\end{lemma}

\begin{proof}
Consider first the stamement about two functions $f$ and $f'$. Initialize $i,i',i'':=1$. While $i \leq L+1$ and $i' \leq L'+1$: 
\begin{enumerate}
\item $\gamma''_{i''}:=\gamma_i + \gamma'_{i'}$;
\item If $\alpha_i < \alpha'_{i'}$, then $\alpha''_{i''} := \alpha_i$,  $i:=i+1$;
\item If $\alpha_i > \alpha'_{i'}$, then $\alpha''_{i''} := \alpha'_{i'}$, $i':=i'+1$;
\item If $\alpha_i=\alpha'_{i'}$, then $\alpha''_{i''} : =\alpha_i$, $i:=i+1$, $i':=i'+1$;
\item $i'':=i''+1$;
\end{enumerate}
Note that $\alpha_{L+1} = \alpha'_{L'+1}=4^N$, so the while loop is executed at most $L+L'+1$ times.  At the end, the arrays $\alpha'', \gamma''$ contain the interval-slope representation of $f''$. 

The statement about a sum of more than two functions follows by induction. 
\end{proof}

In order to use Lemma \ref{lemma:sum_of_piecewise_linear_functions} to compute the CDF of $J$, note that 
\begin{equation}
\prob{J \leq j} = \sum_{q' \in \discretesimplex{\supportv}{n}} \prob{J \leq j, V \in \typeclass{q'}}
\end{equation}
Therefore, it suffices to compute the interval-slope representation of $\prob{J \leq j, V \in \typeclass{q'}}$ as a function of $j$ for each $q' \in \discretesimplex{\supportv}{n}$. This is achieved in the following:

\begin{lemma}\label{lemma:computing_interval_slope_representations}
Consider the following algorithm:
\begin{enumerate}
\item Pre-processing: For each $q \in \discretesimplex{\supportu\times\supportv}{n}$ compute: 
\begin{enumerate}
\item The size of the type class $\typeclass{q} \subset \parenth{\supportu\times\supportv}^n$, namely $|\typeclass{q}| = n! / \parenth{\prod_{(s,t)\in\supportu\times\supportv}(nq(s,t))!}$. 
\item The probability that $U,V$ equals any particular element in $\typeclass{q}$, namely $2^{-n\parenth{H(q) + \relent{q}{p}}} $.
\item The marginal $q'(\cdot) = \sum_{s \in \supportu} q(s,\cdot) \in \discretesimplex{\supportv}{n}$.
\item The size of the type class $\typeclass{q'} \subset \supportv^n$, namely $|\typeclass{q'}|=n! / \parenth{\prod_{t \in \supportv} (nq'(t))!}$. 
\end{enumerate}
Then, sort the elements $q \in \discretesimplex{\supportu\times\supportv}{n}$ in decreasing order according to the probability that $U,V$ equals any particular element of $\typeclass{q}$. 
\item Construction of the interval-slope representations: go over the $q \in \discretesimplex{\supportu\times\supportv}{n}$ in the established order. For each $q$, update the interval-slope representation corresponding to its marginal $q'$. Specifically, append to the interval-slope representation of $\prob{J\leq j,V\in\typeclass{q'}}$ an interval of length $|\typeclass{q}|/|\typeclass{q'}|$ and slope $|\typeclass{q'}|2^{-n\parenth{H(q)+\relent{q}{p}}}$. 
\end{enumerate}
Then, this algorithm runs in time $O\parenth{n^{|\supportu||\supportv|}}$, and outputs the interval-slope representation of $\prob{J\leq j, V\in\typeclass{q'}}$ as a function of $j$ for all $q' \in \discretesimplex{\supportv}{n}$. Moreover, the total number of intervals in all the representations is $O\parenth{n^{|\supportu||\supportv|-1}}$
\end{lemma}

\begin{proof}
First, consider the total number of intervals. One interval is added for each $q \in \discretesimplex{\supportu\times\supportv}{n}$, so the total number of intervals is 
\begin{equation}
\absval{\discretesimplex{\supportu\times\supportv}{n}} = {n+|\supportu\times\supportv|-1 \choose |\supportu\times\supportv|-1}= O\parenth{n^{|\supportu||\supportv|-1}}
\end{equation}
If $|\supportu|^n < 4^N$, then after the completion of the main loop, it is necessary to append the interval $(|S|^n, 4^N)$ with zero slope in order to match the form required by Lemma \ref{lemma:sum_of_piecewise_linear_functions}. This adds $\absval{\discretesimplex{\supportv}{n}}$ intervals.

Next, consider the running time. Computation of type class sizes, marginal and probability of a particular element takes time $O(n)$ for each $q$, so the total time is $O\parenth{n^{|\supportu||\supportv|}}$. Sorting takes time $O\parenth{|\discretesimplex{\supportu\times\supportv}{n}| \log |\discretesimplex{\supportu\times\supportv}{n}| } = O\parenth{n^{|\supportu||\supportv|-1}\log n}$. Construction of the interval-slope representations takes time $O\parenth{n^{|\supportu||\supportv|-1}}$. 

Next, consider the correctness of the algorithm. Consider a three-level partition of $(\supportu\times\supportv)^n$. First,
\begin{equation}\label{eq:first_level_of_partition}
(\supportu\times\supportv)^n = \cup_{q' \in \discretesimplex{\supportv}{n}} \supportu^n \times \typeclass{q'}
\end{equation}
At the second level, for each $q' \in \discretesimplex{\supportv}{n}$, 
\begin{equation}\label{eq:second_level_of_partition}
\supportu^n \times \typeclass{q'} = \cup_{q \in \discretesimplex{\supportu\times\supportv}{n}:\sum_{s\in\supportu} q(s,\cdot)=q'} \typeclass{q}
\end{equation}
At the third level, for each $q \in \discretesimplex{\supportu\times\supportv}{n}$ with marginal $q'\in\discretesimplex{\supportv}{n}$,
\begin{equation}\label{eq:third_level_of_partition}
\typeclass{q} = \cup_{v \in \typeclass{q'}} \parenth{\typeclass{q} \cap \parenth{\supportu^n \times \{v\}}}
\end{equation}
Note that the sizes of the sets on the right-hand-side of \eqref{eq:third_level_of_partition} are all the same. This follows from Lemma \ref{lemma:orbits_of_the_diagonal_action}, because the type classes $\typeclass{q}$ and $\typeclass{q'}$ are orbits of the group of permutations of $n$ elements acting on $\supportu^n \times \supportv^n$ and on $\supportv^n$ respectively. Therefore, for each $q \in \discretesimplex{\supportu\times\supportv}{n}$ with marginal $q' \in \discretesimplex{\supportv}{n}$ and for each $v \in \typeclass{q'}$, 
\begin{equation}\label{eq:joint_type_class_and_particular_v}
\absval{\typeclass{q} \cap\parenth{\supportu^n \times \{v\}}} = \frac{|\typeclass{q}|}{|\typeclass{q'}|}
\end{equation}

Consider now the piecewise linear function $\prob{J\leq j, V=v}$  for any $v$. It depends only on the tuple $\parenth{p_{UV}(u,v)}_{u \in \supportu^n}$. The value $p_{UV}(u,v)=2^{-n\parenth{H(\type{u,v})+\relent{\type{u,v}}{p} } }$ depends only on the joint type $\type{u,v}$ of $u,v$. The number of times that $2^{-n\parenth{H(\type{u,v})+\relent{\type{u,v}}{p} } }$ appears in the tuple $\parenth{p_{UV}(u,v)}_{u \in \supportu^n}$ is $|\typeclass{\type{u,v}}|/|\typeclass{\type{v}}|$, according to \eqref{eq:joint_type_class_and_particular_v}. This has two consequences: 
\begin{enumerate}
\item The interval-slope representation of $\prob{J\leq j, V=v}$ contains intervals of length $|\typeclass{q}|/|\typeclass{\type{v}}|$ and slope $2^{-n\parenth{H(q)+\relent{q}{p}}}$, sorted according to the slope, where $q$ ranges over all joint types whose marginal is $\type{v}$. 
\item If $v,v'$ have the same type, then $\prob{J\leq j,V=v}=\prob{J\leq j, V=v'}$ for all $j$. 
\end{enumerate}
Therefore, for any $q' \in \discretesimplex{\supportv}{n}$ and any $v \in \typeclass{q'}$, 
\begin{equation}\label{eq:cdf_for_type_class_and_cdf_for_particular_element}
\prob{J\leq j, V \in \typeclass{q'}} = |\typeclass{q'}| \prob{J\leq j, V=v}
\end{equation}
so the algorithm outputs the correct interval-slope representations. 
\end{proof}

This completes the proof of Theorem \ref{thm:computing_the_bounds}. In summary:
\begin{enumerate}
\item Lemma \ref{lemma:computing_interval_slope_representations} gives an $O\parenth{n^{|\supportu||\supportv|}}$ time algorithm to compute the interval-slope representations of $\prob{J\leq j, V \in \typeclass{q'}}$ as a function of $j$, for all $q' \in \discretesimplex{\supportv}{n}$. The total number of intervals is $O\parenth{n^{|S||T|-1}}$.
\item Lemma \ref{lemma:sum_of_piecewise_linear_functions} gives an $O\parenth{n^{|\supportu||\supportv|-1}}$ time algorithm to combine these into the interval-slope representation of the CDF of $J$, with $L +1= O\parenth{n^{|\supportu||\supportv|-1}}$ intervals. Lemma \ref{lemma:points_from_interval_slope} shows that the coordinates of the $L$ points where the CDF of $J$ changes slope can be computed in $O\parenth{n^{|\supportu||\supportv|-1}}$ additional time. 
\item Given the coordinates of these $L$ points, Lemma \ref{lemma:computing_the_bounds_from_points} gives an $O(NL)=O\parenth{n^{|\supportu||\supportv|}}$ time algorithm to compute $\errorguesserrorprobconverse{\tilde{p}_{UV},r}$ and $\errorguesserrorprobachievability{\tilde{p}_{UV},r}$ for all $r=k/N$, $k=1,\dots,N$. After these are computed and stored, an $O(\log N)=O(\log n)$ time binary search can be used to compute $\errorguessrateconverse{\tilde{p}_{UV},\epsilon}$ and $\errorguessrateachievability{\tilde{p}_{UV},\epsilon}$ for a given $\epsilon >0$. 
\end{enumerate}

\emph{Remark:} In the special cases of the erasure and depolarizing channels, Theorems \ref{thm:erasure_channel_bounds} and \ref{thm:depolarizing_channel_bounds} give a more efficient way to compute the bounds than Theorem \ref{thm:computing_the_bounds}, because the former take into account the special structure of the erasure and depolarizing channels. 

\section{Conclusion}\label{sec:conclusion}

This article defines a family of groups of lower triangular matrices and uses them to establish a canonical form for unrestricted and stabilizer parity check matrices of a given size and rank. It also shows that the canonical form for the Clifford group can be computed in time $O(n^3)$, which improves upon the previously known time $O(n^6)$. Finally, the present article establishes a finite blocklength refinement of the hashing bound for stabilizer codes and Pauli noise, and shows that this achievability bound is nearly optimal among the class of arguments that use guessing the error as a substitute for guessing the coset. 

A possible direction for future work is to investigate the relation between guessing the error and guessing the coset. The concatenated coding examples of \cite{shor1996quantumerror,divincenzo1998quantumchannel,fern2008correctable,fern2008lowerbounds}
show that very noisy depolarizing channels can have positive capacity even though the hashing bound is non-positive at those values of the channel parameter. It would be interesting to see whether the techniques in the present article can be used to construct further examples of strict separation between guessing the error and guessing the coset for stabilizer codes.


\begin{thebibliography}{10}

\bibitem{gottesman1996class}
D.~Gottesman, ``Class of quantum error-correcting codes saturating the quantum
  hamming bound,'' \emph{Phys. Rev. A}, vol.~54, pp. 1862--1868, Sep 1996.
  [Online]. Available: \url{https://link.aps.org/doi/10.1103/PhysRevA.54.1862}

\bibitem{calderbank1997quantum}
A.~R. Calderbank, E.~M. Rains, P.~W. Shor, and N.~J.~A. Sloane, ``Quantum error
  correction and orthogonal geometry,'' \emph{Phys. Rev. Lett.}, vol.~78, pp.
  405--408, Jan 1997. [Online]. Available:
  \url{https://link.aps.org/doi/10.1103/PhysRevLett.78.405}

\bibitem{duncan2020graphtheoretic}
R.~Duncan, A.~Kissinger, S.~Perdrix, and J.~van~de Wetering, ``Graph-theoretic
  {S}implification of {Q}uantum {C}ircuits with the {ZX}-calculus,''
  \emph{{Quantum}}, vol.~4, p. 279, Jun. 2020. [Online]. Available:
  \url{https://doi.org/10.22331/q-2020-06-04-279}

\bibitem{bravyi2021hadamardfree}
S.~Bravyi and D.~Maslov, ``Hadamard-free circuits expose the structure of the
  clifford group,'' \emph{IEEE Transactions on Information Theory}, vol.~67,
  no.~7, pp. 4546--4563, July 2021. [Online]. Available:
  \url{https://doi.org/10.1109/TIT.2021.3081415}

\bibitem{bassler2023synthesisof}
P.~Ba{\ss{}}ler, M.~Zipper, C.~Cedzich, M.~Heinrich, P.~H. Huber, M.~Johanning,
  and M.~Kliesch, ``Synthesis of and compilation with time-optimal multi-qubit
  gates,'' \emph{{Quantum}}, vol.~7, p. 984, Apr. 2023. [Online]. Available:
  \url{https://doi.org/10.22331/q-2023-04-20-984}

\bibitem{khesin2025universalgraphrepresentation}
A.~B. Khesin, J.~Z. Lu, and P.~W. Shor, ``Universal graph representation of
  stabilizer codes,'' \emph{PRX Quantum}, vol.~6, p. 040325, Nov 2025.
  [Online]. Available: \url{https://link.aps.org/doi/10.1103/1gjs-2rhx}

\bibitem{ostrev2024quantum}
D.~Ostrev, ``Quantum ldpc codes from intersecting subsets,'' \emph{IEEE
  Transactions on Information Theory}, vol.~70, no.~8, pp. 5692--5709, Aug
  2024. [Online]. Available:
  \url{https://doi.org/10.1109/TIT.2024.3402091}

\bibitem{nielsen2010quantum}
M.~A. Nielsen and I.~L. Chuang, \emph{Quantum Computation and Quantum
  Information}.\hskip 1em plus 0.5em minus 0.4em\relax Cambridge University
  Press, jun 2010. [Online]. Available:
  \url{https://doi.org/10.1017/cbo9780511976667}

\bibitem{maslov2018shorterstabilizer}
D.~Maslov and M.~Roetteler, ``Shorter stabilizer circuits via bruhat
  decomposition and quantum circuit transformations,'' \emph{IEEE Transactions
  on Information Theory}, vol.~64, no.~7, pp. 4729--4738, 2018. [Online]. Available: \url{https://doi.org/10.1109/TIT.2018.2825602}

\bibitem{polyanskiy2010channel}
Y.~Polyanskiy, H.~V. Poor, and S.~Verdu, ``Channel coding rate in the finite
  blocklength regime,'' \emph{IEEE Transactions on Information Theory},
  vol.~56, no.~5, pp. 2307--2359, 2010. [Online]. Available: \url{https://doi.org/10.1109/TIT.2010.2043769}

\bibitem{ashikhmin2014fidelitylowerbounds}
A.~Ashikhmin, ``Fidelity lower bounds for stabilizer and css quantum codes,''
  \emph{IEEE Transactions on Information Theory}, vol.~60, no.~6, pp.
  3104--3116, June 2014. [Online]. Available: \url{https://doi.org/10.1109/TIT.2014.2303801}

\bibitem{tomamichel2016quantum}
M.~Tomamichel, M.~Berta, and J.~M. Renes, ``Quantum coding with finite
  resources,'' \emph{Nature Communications}, vol.~7, no.~1, may 2016. [Online].
  Available: \url{https://doi.org/10.1038\%2Fncomms11419}

\bibitem{wilde2017quantum}
M.~M. Wilde, \emph{Quantum Information Theory}, 2nd~ed.\hskip 1em plus 0.5em
  minus 0.4em\relax Cambridge University Press, 2017. [Online]. Available: \url{https://doi.org/10.1017/9781316809976}

\bibitem{strang2015algebra}
G.~Strang, ``The algebra of elimination,'' in \emph{Excursions in Harmonic
  Analysis, Volume 3: The February Fourier Talks at the Norbert Wiener
  Center}.\hskip 1em plus 0.5em minus 0.4em\relax Springer, 2015, pp. 3--22. [Online]. Available: \url{https://doi.org/10.1007/978-3-319-13230-3_1}

\bibitem{bennett1996mixed}
C.~H. Bennett, D.~P. DiVincenzo, J.~A. Smolin, and W.~K. Wootters,
  ``Mixed-state entanglement and quantum error correction,'' \emph{Phys. Rev.
  A}, vol.~54, pp. 3824--3851, Nov 1996. [Online]. Available:
  \url{https://link.aps.org/doi/10.1103/PhysRevA.54.3824}

\bibitem{berry1941accuracy}
A.~C. Berry, ``The accuracy of the {G}aussian approximation to the sum of
  independent variates,'' \emph{Trans. Amer. Math. Soc.}, vol.~49, pp.
  122--136, 1941. [Online]. Available: \url{https://doi.org/10.2307/1990053}

\bibitem{esseen1942liapounoff}
C.-G. Esseen, ``On the {L}iapounoff limit of error in the theory of
  probability,'' \emph{Ark. Mat. Astr. Fys.}, vol. 28A,, no.~9, p.~19, 1942.

\bibitem{macwilliams1977theory}
F.~J. MacWilliams and N.~J.~A. Sloane, \emph{The theory of error-correcting
  codes}.\hskip 1em plus 0.5em minus 0.4em\relax Elsevier, 1977, vol.~16. [Online]. Available: \url{https://www.sciencedirect.com/bookseries/north-holland-mathematical-library/vol/16/suppl/C}

\bibitem{shor1996quantumerror}
P.~W. {Shor} and J.~A. {Smolin}, ``{Quantum Error-Correcting Codes Need Not
  Completely Reveal the Error Syndrome},'' \emph{arXiv e-prints}, pp.
  quant--ph/9\,604\,006, Apr. 1996. [Online]. Available: \url{https://doi.org/10.48550/arXiv.quant-ph/9604006}

\bibitem{divincenzo1998quantumchannel}
D.~P. DiVincenzo, P.~W. Shor, and J.~A. Smolin, ``Quantum-channel capacity of
  very noisy channels,'' \emph{Phys. Rev. A}, vol.~57, pp. 830--839, Feb 1998.
  [Online]. Available: \url{https://link.aps.org/doi/10.1103/PhysRevA.57.830}

\bibitem{fern2008correctable}
J.~Fern, ``Correctable noise of quantum-error-correcting codes under adaptive
  concatenation,'' \emph{Phys. Rev. A}, vol.~77, p. 010301, Jan 2008. [Online].
  Available: \url{https://link.aps.org/doi/10.1103/PhysRevA.77.010301}

\bibitem{fern2008lowerbounds}
J.~Fern and K.~B. Whaley, ``Lower bounds on the nonzero capacity of pauli
  channels,'' \emph{Phys. Rev. A}, vol.~78, p. 062335, Dec 2008. [Online].
  Available: \url{https://link.aps.org/doi/10.1103/PhysRevA.78.062335}

\end{thebibliography}
\end{document}